\documentclass[preprint,12pt]{elsarticle}

 \usepackage[margin=2.5cm]{geometry}

\usepackage{amssymb}

\usepackage{amsmath,amsthm}
\usepackage{amssymb, amsbsy}
\usepackage{gensymb}

\usepackage{graphicx}

\usepackage{import}
\usepackage{enumitem, outlines}
\usepackage{lineno}

\usepackage[edges]{forest}
\usepackage{amsmath}
\usepackage{amssymb}

\newcommand{\MED}{$\mathcal{MED}(\mathcal R)$}
\newcommand{\Diam}{$\mathcal{M}$}
\newcommand{\Diamo}{$\mathcal{M}'$}
\newcommand{\ED}{$\mathcal{ED}(R)$}
\newcommand{\CNTR}{$\mathcal{C}(\mathcal{M})$}
\newcommand{\Bd}[1]{\mathcal{P}(\mathcal{#1})}

\usepackage{setspace}
\usepackage[list=true]{subcaption}
\usepackage{dcolumn, multirow}
\usepackage{array}

\usepackage{doi}

\newtheorem{theorem}{Theorem}
\newtheorem{lemma}{Lemma}
\newtheorem{corollary}{Corollary}
\newtheorem{definition}{Definition}
\newtheorem{observation}{Observation}

\journal{}

\begin{document}

\begin{frontmatter}



\title{Min-Max Gathering on Infinite Grid}

\author[inst1]{Abhinav Chakraborty}
\ead{abhinav.chakraborty06@gmail.com}

\author[inst2]{Pritam Goswami \corref{mycorrespondingauthor}}
 
\cortext[mycorrespondingauthor]{Corresponding Author: Pritam Goswami}
\ead{pgoswami.cs@gmail.com}

\author[inst3]{Satakshi Ghosh}
\ead{satakshi.ghos@gmail.com }

\affiliation[inst1]{organization={Department of Mathematics, Birla Institute of Technology, Mesra, Ranchi-835215, Jharkhand, India}}

\affiliation[inst2]{organization={Department of Computer Science and Engineering, Sister Nivedita University, Kolkata, India}}

\affiliation[inst3]{organization={Department of Basic Science and Humanities - International Institute of Information Technology, Bhubaneswar, India}}

\begin{abstract}
Gathering is a fundamental coordination problem in swarm robotics, where the objective is to bring robots together at a point not known to them at the beginning. While most research focuses on continuous domains, some studies also examine the discrete domain. This paper addresses the optimal gathering problem on an infinite grid, aiming to improve the energy efficiency by minimizing the maximum distance any robot must travel. The robots are autonomous, anonymous, homogeneous, identical, and oblivious. We identify all initial configurations where the optimal gathering problem is unsolvable. For the remaining configurations, we introduce a deterministic distributed algorithm that effectively gathers $n$ robots ($n\ge 9$). The algorithm ensures that the robots gathers at one of the designated min-max nodes in the grid. Additionally, we provide a comprehensive characterization of the subgraph formed by the min-max nodes in this infinite grid model.
\end{abstract}




\end{frontmatter}





\section{Introduction}
\textit{Swarm robotics} is a branch of robotics that draws inspiration from the behavior of social insects, such as bees, ants, and termites. In contrast to traditional robotics, robot swarms are distributed multi-robot systems consisting of simple and inexpensive robots that can collaboratively execute complex tasks, such as forming patterns, converging at a single point, exploring unknown environments, etc. Indeed, the coordination of autonomous robot swarms is a multidisciplinary field of study that encompasses various domains, including robotics, artificial intelligence, control theory, and distributed computing. Within the realm of distributed computing, researchers have been particularly interested in developing algorithms and methods for efficiently coordinating and controlling robot swarms. 

In the traditional model, the robots are generally assumed to be \textit{anonymous} (no unique identifiers), \textit{autonomous} (no central coordination), \textit{homogeneous} (each robot executes the same deterministic algorithm), \textit{identical} (indistinguishable by their appearances) and \textit{oblivious} (no persistent memory of the past events). Generally, the robots are assumed to be points in the Euclidean plane. In spite of this, there are numerous studies that consider the deployment region of the robots to be an undirected graph. In this paper, the deployment region of the robots is assumed to be an infinite grid graph, which may be considered as a natural discretization of the plane. The robots are assumed to be \textit{disoriented} (the robots neither have any agreement on the common coordinate system, nor do they have any agreement on a common chirality). They do not have any access to a \textit{global coordinate system}. They do not have any explicit means of communication, i.e., they cannot send any messages to other robots. Each robot has its own local coordinate system with the assumption that the origin is the current position of the robot.

A robot is either active or inactive. When a robot is active, it operates in the \textit{Look-Compute-Move} cycle. In the \textit{look} phase, a robot takes a snapshot of the current configuration in terms of its own local coordinate system. In the \textit{Compute} phase, a robot decides to stay idle or to move to one of its adjacent nodes. Finally, a robot moves towards its destination in the \textit{Move} phase, where the destination node is computed in the Compute phase. If the destination of the robot computed in the Compute phase is its current position, then it performs a null movement. In the existing literature, three types of schedulers are described based on the timing and activation of the robots. In the \textit{fully synchronous} (FSYNC) setting, all the robots are activated simultaneously. The activation phase of all robots is divided into global rounds, which means that all robots perform their actions and computations during each round. In the \textit{semi-synchronous} (SSYNC) setting, not all robots are necessarily activated simultaneously in each round. A subset of robots may be activated together, while others may remain idle in a given round. The most realistic and general model is the \textit{asynchronous} setting, where there is no common notion of time. The duration of the Look, Compute, and Move phases for each robot is finite but unpredictable and is decided by the adversary for each robot. In this paper, we have assumed the scheduler to be fair asynchronous, which means that each robot will eventually perform its $LCM$ cycle within a finite time, ensuring that no robot is unfairly delayed indefinitely.

As the deployment region of the robots is assumed to be an infinite grid graph, a robot can move toward one of its adjacent nodes at any point of time. The movement of the robot is assumed to be instantaneous, i.e., each robot can be observed only at the nodes of the graph. The movement is restricted only along the edges of the graph. During the execution of any algorithm, a node may contain multiple robots. When such a situation occurs, we say that a \textit{multiplicity} exists on that particular node. In this paper, we have assumed that the robots are equipped with \textit{global-weak multiplicity detection capability}. With this capability, the robots can detect whether a node is occupied by any robot multiplicity.

In this paper, we have studied the min-max gathering problem in a scenario where the robots are deployed at the nodes of an infinite grid, given as an input. In this gathering problem, the optimization constraint is to minimize the maximum distance traveled by a robot during the gathering. Bhagat et al. \cite{DBLP:conf/spaa/BhagatM21,DBLP:conf/caldam/BhagatM17} were the first to address the same problem in the continuous domain. To the best of our knowledge, this is the first paper that addresses the following optimal gathering problem in infinite grids. In the continuous domain, there is exactly a unique min-max point, which is the center of the minimum enclosing circle of all the robots. Nevertheless, we have shown that the cardinality of the set of min-max nodes may be greater than one in the infinite grid case. In the later section of the paper, we discuss examples where multiple min-max nodes can exist in a given configuration. The existence of multiple min-max nodes in configurations could indeed make the gathering problem more challenging in this scenario.
\subsection{Motivation}
This paper aims to study the constrained optimal gathering problem in the discrete domain, where the robots are deployed at the nodes of an infinite grid, given as an input. Previous research has mostly tackled the gathering problem in the continuous domain, assuming that robots are represented as points and can make infinitesimal movements. In the continuous domain, it is assumed that the robots move with high accuracy and infinite precision. In some continuous models, robots can perform guided movements along specified curves \cite{DBLP:journals/jpdc/PattanayakMRM19,DBLP:journals/dc/CiceroneSN19}. This implies a high level of control and precision in their motion. Real-life robots have limitations in terms of mechanical capabilities and sensor precision. They cannot move with infinite precision and with such high accuracy. To address these real-world limitations, this paper focuses on a grid-based terrain. In this representation, robots can only move along grid lines and can move to one of their neighboring grid nodes in a single step. This grid-based approach is more realistic for actual robot deployments. 

Another motivation for considering this problem is theoretical. In this paper, we have assumed the constrained optimal gathering problem, where the optimal constraint is to minimize the maximum distance traveled by any robot. This problem was first considered by Bhagat et al. \cite{DBLP:conf/spaa/BhagatM21,DBLP:conf/caldam/BhagatM17} in the continuous domain. To the best of our knowledge, we are the first to study this problem in the discrete domain. In the continuous domain, there is exactly a unique min-max point. However, in the infinite grid case, we have shown that the cardinality of the set of min-max nodes can be more than one. The presence of multiple min-max nodes in configurations can indeed complicate the gathering process. Since robots are often constrained by the availability of limited energy resources (e.g., batteries), minimizing their travel distance can enhance their operational efficiency and prolong their lifespan. It is important for the robots to find pathways or trajectories that will enable them to reach the target location while avoiding unnecessary detours.
\subsection{Related Works}
Numerous studies and analyses have been conducted on the gathering and rendezvous problem, providing insight into the possibilities and impossibilities of various assumptions. The problems have been extensively studied in the continuous domain under different schedulers and under different assumptions of agreement in coordinate axes \cite{10.1007/3-540-45061-0_90,PRENCIPE2007222,doi:10.1137/050645221,doi:10.1137/100796534,bouzid2013gathering,prencipe2007impossibility,izumi2012gathering,flocchini2005gathering}. In the continuous domain, the robots are deployed in the Euclidean plane. However, in the discrete domain, the robots are usually deployed at the nodes of an undirected and unlabeled graph. In the discrete domain, a significant portion of research has indeed focused on specific topological structures, such as rings \cite{DBLP:journals/tcs/KlasingMP08,DBLP:journals/tcs/KlasingKN10,DBLP:journals/jda/DAngeloSN14,DBLP:journals/dc/DAngeloSN14}. In \cite{DBLP:journals/tcs/KlasingMP08}, the authors studied the gathering problem in an anonymous ring and proved that the gathering is impossible in a ring without any multiplicity detection capability of the robots. A distributed deterministic algorithm was proposed for all asymmetric and symmetric configurations with an odd number of robots. The proposed algorithms assume the global-weak multiplicity detection capability of the robots. The paper provides impossibility results by proving that gathering is unsolvable in certain configurations. Specifically, they have shown that gathering is not achievable for periodic configurations, which are characterized by rotations and for configurations that admit edge-edge symmetries. Edge-edge symmetries imply that the axis of symmetry passes through two edges of the ring. Klasing et al. \cite{DBLP:journals/tcs/KlasingKN10} investigated configurations in an anonymous ring that allows symmetries and has an
even number of robots. The robots are endowed with global-weak multiplicity detection capability. They solved the problem for all configurations with more than eighteen
robots.

In the discrete domain, the gathering problem has also been studied in finite \cite{DBLP:journals/tcs/DAngeloSKN16,info12110448,CASTENOW2020289} and infinite grids \cite{DBLP:journals/iandc/StefanoN17,DBLP:conf/icara/GhoshSGS23,DBLP:journals/corr/abs-2204-14042,DBLP:journals/ijfcs/BhagatCDM23,DBLP:journals/fuin/BhagatCDM22}. D'Angelo et al. \cite{DBLP:journals/tcs/DAngeloSKN16} studied the gathering problem in finite grids. They proved that even with global-strong multiplicity detection capability, a configuration is ungatherable if and only if it is periodic or symmetric, with the line of symmetry passing through the edges of the grid. They solved the problem for the remaining configurations without assuming any multiplicity detection capability of the robots. Castenow et al. \cite{CASTENOW2020289} proposed an algorithm for the gathering problem on a grid in the fully synchronous setting for oblivious robots with limited visibility. Poudel et al. \cite{info12110448} studied the problem in finite grids, where the robots agree on the direction and orientation of both the coordinate axes and have limited viewing range. In this paper, a distributed time-optimal algorithm was proposed in the asynchronous setting. In infinite grids, Di Stefano et al. \cite{DBLP:journals/iandc/StefanoN17} studied the optimal gathering problem, where the objective constraint is to minimize the sum of the distances traveled by all the robots in order to accomplish the gathering. The authors proposed a distributed algorithm assuming that the robots have global-strong multiplicity detection capability and assured gathering on a Weber point by letting each robot move along the shortest paths toward such a node. A Weber point is a node of the graph that minimizes the sum of the lengths of the shortest paths from it to each robot. A series of papers concerns the gathering over meeting nodes problem in infinite grids \cite{DBLP:journals/ijfcs/BhagatCDM23,DBLP:journals/fuin/BhagatCDM22}, where the gathering is constrained to a specific set of predefined meeting nodes that are situated at the nodes of an infinite grid. In \cite{DBLP:journals/ijfcs/BhagatCDM23}, the optimal gathering over meeting nodes has been studied, where the objective constraint is to minimize the sum of the total distances traveled by each robot. Shibata et al. \cite{DBLP:journals/ijnc/ShibataOSNKK22} studied the gathering problem of seven robots in infinite triangular grids, where the proposed algorithm is optimal with respect to the visibility range. The robots agree on
the orientation and direction of both axes and they operate under a fully synchronous
scheduler. Goswami et al.\cite{DBLP:journals/corr/abs-2204-14042} studied the gathering problem on an infinite triangular grid for $n \geq 2$ robots, where the robots have limited visibility.

\noindent Cicerone et al. \cite{DBLP:journals/dc/CiceroneSN18} considered a variation of gathering problem in the continuous domain, where the robots need to gather at one of the meeting points, which are some prefixed meeting points on the plane. In this paper, they considered two different objective functions, one is to minimize the sum of the distances traveled by each robot in order to finalize the gathering at a meeting point and the other is to minimize the maximum distance traveled by any robot in order to accomplish the gathering. Bhagat et al. \cite{DBLP:conf/spaa/BhagatM21,DBLP:conf/caldam/BhagatM17} was the first to study the min-max gathering of oblivious robots in the continuous domain. The objective is to minimize the maximum distance traveled by any robot. In \cite{DBLP:conf/caldam/BhagatM17}, it was shown that even in the $\mathcal FSYNC$ model, the constrained optimal gathering problem is not solvable for oblivious robots even with multiplicity detection capability. A distributed algorithm was proposed to solve the problem for a set of $n \geq 5$ asynchronous robots under a persistent memory model, which is implemented by endowing the robots with externally visible lights. There are a finite number of colors that can be assigned to these lights. The algorithm uses only two bits of persistent memory and solves the problem in finite time. In \cite{DBLP:conf/spaa/BhagatM21}, the min-max gathering problem was studied for the oblivious robots, where the robots have local-weak multiplicity detection capability and one-axis agreement. The number of robots was assumed to be at least 6. However, in the case when the robots are deployed on the plane, there is a unique min-max point, which is the center of the minimum enclosing circle of all the robots. In case, when the robots are deployed at the nodes of an infinite grid, the configuration may have multiple min-max nodes. 

\subsection{Our Contributions}
In this paper, we have considered the min-max gathering problem of oblivious robots in infinite grids. The objective of this problem is to minimize the maximum distance traveled by any robot. In the continuous domain, the min-max points defined by the points that minimize the maximum distance traveled by any robot are unique. However, we have shown example configurations where multiple min-max nodes may exist in case the robots are deployed at the nodes of an infinite grid. In this paper, we have shown a complete characterization of the min-max nodes in an infinite grid model. We have characterized configurations for which the problem remains unsolvable, and for the rest of the configurations, a deterministic distributed algorithm has been proposed that ensures the gathering of the robots at one of the min-max nodes for $n \geq 9$.
\subsection{Outline}
The following section, section \ref{sec:models}, describes the robot model and some definitions that are relevant in understanding the problem. Section \ref{sec:problem} describes the problem definition that has been considered in this paper. In section \ref{sec:strategy}, a full characterization of the graph induced by the set of min-max nodes is proposed, and the configuration which are ungatherable has been investigated. Section \ref{sec:algorithm} proposes a deterministic distributed algorithm that ensures the gathering of the robots at one of the min-max nodes. In section \ref{sec:correctness}, the correctness of the algorithm has been described. Finally, in section \ref{sec:conclusion}, we conclude the paper with some future directions to work with. 
\section{Models and Definitions} \label{sec:models}
\subsection{Models}
\begin{itemize}
    \item \textit{Robots:} We have assumed the following robot model:
    \begin{itemize}
    \item \textit{Dimensionless:} The robots are modelled as nodes in the input grid.
    \item \textit{Anonymous:} They have no unique identifiers.
    \item \textit{Autonomous:} There is no central coordination, i.e., the robots operate independently, making decisions and taking actions without any external control.
    \item \textit{Identical:} They are indistinguishable by their appearance.
    \item \textit{Homogeneous}: All robots execute the same deterministic algorithm.
		\item \textit{Oblivious}: The robots do not have a persistent memory of past information.
		\item \textit{Silent}: There is no explicit means of direct communication.
		\item \textit{Disoriented}: The robots do not have access to a \textit{global coordinate system}. There is no common compass and no agreement on chirality.
		\item \textit{Unlimited visibility}: They can perceive the entire graph.
    
\end{itemize}
\item \textit{Look-Compute-Move Cycle:} A robot is either active or in an inactive state. If a robot is active, it operates in \textit{Look-Compute-Move} cycle. In the Look phase, a robot obtains a snapshot of the current configuration and observes the positions of all the other robots. Based on the current configuration, each robot computes a destination node in the Compute phase. The destination node could be the current position of the robot or a new position computed according to a deterministic distributed algorithm. If the destination node computed is different from its current position, the robot moves towards that destination node in the Move phase. Otherwise, it remains stationary, i.e., in that case, it performs a \textit{null movement}. At the end of each LCM cycle, all local memory of each robot are erased.
\item \textit{Asynchronous Scheduler:} In this paper, the scheduler is assumed to be a fair asynchronous ($\mathcal {ASYNC}$). In the $\mathcal{ASYNC}$ model, there is no synchronized global clock or common notion of time. Each robot operates independently, making decisions based on its own internal clock or other means of timing. The adversary determines the duration of the Look, Compute and Move phases, which are finite but unpredictable. The robots are visible to each other while moving. This means that other robots during the Look phase can observe a robot's movements in real time, even if computations were made using outdated information about their locations. The configuration perceived by a robot during the Look phase may significantly change before it makes a move. 
\item \textit{Visibility:} The visibility of the robots is assumed to be global, i.e., each robot can perceive the entire input grid graph and can observe the positions of all the other robots.
\item \textit{Movement:} The movement of the robot is assumed to be instantaneous, i.e., any robot performing a Look operation observes all the other robot’s positions only at the nodes of the input grid graph and not on the edges. The movement of the robots is restricted only along the edges of the graph and a robot can move towards one of its four adjacent nodes.
\end{itemize}
\subsection{Terminologies}
\begin{itemize}
\item \textit{Notations:}
    \begin{itemize}
        \item $P=(\mathbb{Z}$, $E')$: infinite path graph where the vertex or node set corresponds to the set of integers $\mathbb {Z}$ and the edge set $E'$ consist of the ordered pairs $\lbrace (i$, $i+1)|$ $ i\in \mathbb{Z} \rbrace$. Essentially, this is an infinite sequence of vertices connected by edges in a linear fashion.
			\item \textit{Cartesian product of the graph $P\times P$}: input grid graph.
			\item $V$ and $E$: represent the set of nodes (vertices) and edges, respectively, of the input grid graph.
			\item $d_M(u,v)$: \textit{Manhattan} distance between the nodes $u$ and $v$.
			\item $\mathcal R=\lbrace r_1,r_2,\ldots,r_n\rbrace$: a set of robots deployed at the nodes of the input grid graph.
   \item $r_i(t):$ represents the position of the robot $r_i$ at time $t$. In other words, it is a function that maps time $t$ to the position occupied by robot $r_i$ at that time. When there's no ambiguity, $r_i$ represents both the robot and its position.
   \item $\mathcal R(t)=\lbrace r_1(t),r_2(t),...,r_n(t)\rbrace$: multiset of robot positions at time $t$. At $t=0$, i.e., in the initial configuration, the positions of all robots are distinct. However, at $t> 0$, $r_i(t)$ may be equal to $r_j(t)$, for some $r_i(t), r_j(t) \in \mathcal R(t)$. 
   
    \end{itemize}
    \item \textit{Minimum Enclosing rectangle:} Let $\mathcal{MER}$ denote the minimum enclosing rectangle of all robots, i.e., $\mathcal{MER}$ is the smallest grid-aligned rectangle that contains all the robots. Assume that the dimension of $\mathcal{MER}$ is $a \times b$. Without any loss of generality, we assume that $a \geq b$. The number of grid edges on a side of $\mathcal{MER}$ is used to define its length. 
    \item \textit{System Configuration:} Consider the function $\lambda_t: V\rightarrow\lbrace 0, 1, 2 \rbrace$ defined at any time $t \geq 0$ by
    \[ \lambda_t(v)=\begin{cases} 
0 & \text{if} \;v \textrm { $\notin \mathcal R(t)$ i.e., $v$ is an empty node} \\
1 & \text{if} \;v \textrm{ $\in \mathcal R(t)$ and $v$ is a single robot position}\\
2 & \text{if} \;v \textrm{ $\in \mathcal R(t)$ and $v$ contains multiplicity }\\

\end{cases}
\]
Note that in the initial configuration, we have assumed that the robots are located at distinct nodes of the graph, i.e., at $t=0$, $\lambda_{0} \in \lbrace 0, 1 \rbrace$. The ordered pair $C(t)=(G, \lambda_t)$ denotes the system configuration at time $t$. Without any ambiguity, we will denote $\lambda_t$ by $\lambda$.
    \item \textit{Symmetry:} Consider the function $\lambda$ defined from the set of nodes $V$ to the set $\lbrace 0, 1, 2 \rbrace$. An \textit{automorphism} of a graph $G=(V$, $E)$ is a bijective map $\phi:V\rightarrow V$ that preserves the adjacency structure of the graph. In other words, $u$ and $v$ are adjacent if and only if $\phi (u)$ and $\phi (v)$ are adjacent. An \textit{automorphism of a configuration}, which extends the concept of automorphisms of graphs, preserves adjacency relationships between the nodes while respecting the labels assigned by the function $\lambda$. In other words, an automorphism of a configuration $C(t)$ is an automorphism $\phi$ of the input grid graph such that $\lambda(v)=\lambda(\phi(v))$, for all $v\in V$. Therefore, an automorphism of a configuration maps a free node to a free node, a single robot position to a single robot position and a multiplicity to a multiplicity. Indeed, the set of all automorphisms of a configuration forms a group, denoted by $Aut (C(t))$. If $|Aut( C(t))|=1$, then the configuration is said to be \textit{asymmetric}. Otherwise, it is said to be \textit{symmetric}. As the grid is embedded in the Cartesian plane, a grid can admit only three types of symmetries, viz; \textit{translation}, \textit{reflective} and \textit{rotational}. Due to the presence of a finite number of robots, a translational symmetry is not possible.

    \begin{itemize}
        \item  A reflection symmetry is characterized by a line of symmetry. A line of symmetry can be either horizontal, vertical, or diagonal lines. These lines of symmetry can pass through either the nodes or the edges of the graph. 
    \item A rotational symmetry is characterized by a center of rotational symmetry and by an angle of rotation. The center of rotation can be a node, the center of an edge, or the center of a unit square in the grid. The angle of rotation can be either $90^{\circ}$ or $180^{\circ}$.
    \end{itemize}
    
  \item \textit{Configuration View\cite{DBLP:journals/ijfcs/BhagatCDM23}:}Without loss of generality, assume that the corners of $\mathcal {MER}$ are denoted by $A$, $B$, $C$ and $D$. We are further assuming that $a \geq b$, where $|AB|=a$ and $|AD|=b$. We associate a ternary string with each corner of $\mathcal {MER}$. For example, the ternary string $s_{AD}(t)$, associated to the corner $A$ is defined by scanning the grid lines parallel to the side $AD$ and associating the function $\lambda(v)$ with each node $v$ encountered. In case $\mathcal {MER}$ is a non-square rectangle, i.e., $p >q$, the direction parallel to the smaller side $AD$ is defined as the \textit{string direction} for the corner $A$. In case $\mathcal {MER}$ is a square, the string direction is determined by the larger lexicographic string among $s_{AB}(t)$ and $s_{AD}(t)$. If the configuration is asymmetric, there will always be a unique largest lexicographic string among all the strings associated with the corners of the rectangle. Let $s_i$ be the largest lexicographic string among all the strings associated with the corners of $\mathcal {MER}$. The corner corresponding to this largest string is referred to as the \textit{key corner}. Any corner that is not the key corner is referred to as a \textit{non-key corner}. The \textit{configuration view} of a node is defined as the tuple ($d', x$), where $d'$ denotes the distance of a node from the {\it key corner} in the {\it string direction} and $x$ denote the type of the node, i.e., $x$ denote whether the node is an empty node, contains a single robot position or a multiplicity. Without loss of generality, assume that the configuration is asymmetric and $A$ is the key corner. Let $s_{AD}$ be the string direction associated with the corner $A$. Once we have a unique key corner, then the robots can agree on a common coordinate axes. $A$ is taken as the origin, $AB$ is taken as the $x-$ axis and $AD$ is taken as the $y-$ axis. In Figure \ref{figure1}, $s_{CB}$ is the largest lexicographic string and $s_{CB}=01000100010100100000010001000101000$. Here, $C$ is the key corner, and $CB$ is the string direction, as $|CB| < |CD|$. Suppose that the configuration is symmetric due to an automorphism $\phi$. In such a scenario, a robot $r$ can not distinguish its position $r(t)$ from the position of $\phi (r)(t)$. As a consequence, no algorithm can distinguish $r$ from $\phi(r)$. We denote the robots $r$ and $\phi(r)$ to be \textit{equivalent} robots.

\end{itemize}
 \begin{figure}[h]
				\centering
			{
				\includegraphics[width=0.297\columnwidth]{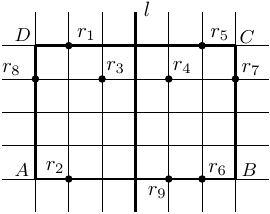}
			}
			\caption{Figure demonstrating the largest lexicographic strings}
			\label{figure1}
		\end{figure}

 

  \section{Problem Definition} \label{sec:problem}

We first consider the following definitions. These definitions are relevant in understanding the problem definition and the concept of min-max nodes.
\begin{definition}[Grid lines passing through a node ($L_H(v)$ and $L_V(v)$)]
\label{def: L_H(v) and L_V(v)}
For a grid node $v = (x_v,y_v)$, the grid lines passing through $v$ are the lines $Y=y_v$ and $X=x_v$ and are denoted by $L_H(v)$ and $L_V(v)$.
    
\end{definition}
\begin{definition}[Diagonals of a grid node ($D(v)$)]
\label{def: D(v)}
For a grid node $v \in V$, the set $D(V) = \{v' \in V : d_m(v',v) =2$ and $ L_H(v),L_V(v) \ne L_H(v'), L_V(v')\}$ are called the diagonals of $v$. 
\end{definition}
\begin{definition}[Neighbours of a grid node ($N(v)$)]
    \label{def:N(v)}
    For a grid node $v \in V$, the set defined by $N(v) = \{v' \in V : d_m(v,v') =1\}$ is denoted as the neighbours of $v$.
\end{definition}
\begin{definition}[Diamond]
A diamond is a square of side length $k\sqrt{2}$ ($k \in \mathbb{N}$) that is formed from a square after a rotation of the square at an angle of $\frac{\pi}{4}$. In addition, the corners of which, after rotation, are grid nodes. The center of a diamond \Diam, denoted by $\mathcal{C}($\Diam$)$, is defined as the intersection of two diagonals of \Diam. We say that the size of diamond $\mathcal{M}$ is the length of its sides, which is of the form $k \sqrt{2}$. The size of a diamond $\mathcal{M}$ is denoted by  $S$(\Diam)
\end{definition} 

\begin{definition}[Perimeter of a diamond \Diam\  ($\Bd{M}$)]
    \label{def: BD(M)}
    For a diamond \Diam\ of size $k\sqrt{2}$ ($k \in \mathbb{N}$), its perimeter is the set of grid nodes $$\Bd{M} = \{v \in V: d_m(v,\mathcal{C}(\mathcal{M})) = k\}.$$
    
\end{definition}
Moreover, all nodes in a diamond's perimeter except for its four corners can be partitioned into four disjoint classes, $\mathcal{P}_{UR}$ (contains the perimeter nodes that are on the upper right of $\mathcal{C}($\Diam$)$), $\mathcal{P}_{UL}$ (contains the perimeter nodes that are on the upper left of $\mathcal{C}($\Diam$)$), $\mathcal{P}_{DL}$ (contains the perimeter nodes that are on the down left of $\mathcal{C}($\Diam$)$) and $\mathcal{P}_{DR}$ (contains the perimeter nodes that are on the down right of $\mathcal{C}($\Diam$)$). Assume that the corner nodes are denoted by $v_{c_1}$, $v_{c_2}$, $v_{c_3}$ and  $v_{c_4}$. We also assume that $v_{c_1}$ is on $L_H(\mathcal{C}(\mathcal{M}))$ and on the right of $\mathcal{C}($\Diam$)$, $v_{c_3}$ is on $L_H(\mathcal{C}(\mathcal{M}))$ on the left of $\mathcal{C}($\Diam$)$, $v_{c_2}$ is on $L_V(\mathcal{C}(\mathcal{M}))$ above $\mathcal{C}($\Diam$)$ and $v_{c_4}$ is on $L_V(\mathcal{C}(\mathcal{M}))$ below $\mathcal{C}($\Diam$)$. Next, we construct the sets $\mathcal{B}_{UR}, \mathcal{B}_{DR}, \mathcal{B}_{UL}, \mathcal{B}_{DL}$ such that, $\mathcal{B}_{UR} = \mathcal{P}_{UR} \cup \{v_{c_1}, v_{c_2}\}$, $\mathcal{B}_{UL} = \mathcal{P}_{UL} \cup \{v_{c_3}, v_{c_2}\}$, $\mathcal{B}_{DL} = \mathcal{P}_{DL} \cup \{v_{c_3}, v_{c_4}\}$ and $\mathcal{B}_{DR} = \mathcal{P}_{DR} \cup \{v_{c_1}, v_{c_4}\}$. Each of these four sets $\mathcal{B}_{UR}, \mathcal{B}_{DR}, \mathcal{B}_{UL}, \mathcal{B}_{DL}$ are called the boundaries of $\mathcal{M}$. For a pair of boundaries, we call the pair \textit{adjacent} if their intersection is non-empty; otherwise, we call them \textit{opposite}. Note that these directions (right, left, up, down) can be defined according to the string directions corresponding to the corners of $\mathcal{MER}$.
\begin{definition}[Enclosing Diamond (\ED)]
   A diamond \Diam\ is called an enclosing diamond for $\mathcal R$ or, \ED\ if all robots in $\mathcal R$ are either on the boundary or strictly inside of \Diam.
\end{definition}
\begin{definition}[Minimal Enclosing Diamond(\MED)]
   A diamond in \ED, say \Diam,  is called a minimal enclosing diamond for $\mathcal R$, or \MED \  if and only if for any other \Diam$' \in$ \ED, $S$(\Diam$'$) $\ge S$(\Diam).
\end{definition}
\begin{definition}[Min-Max Node]
    We say a grid node $v$ is a min-max node if and only if $\underset{u \in \mathcal R}{\max}\  d_m(u,v) \le \underset{u \in \mathcal R}{\max}\  d_m(u,v')$, for all $v' \in V\setminus \{v\}$. Furthermore, we define the set of all min-max nodes for a configuration $C(t)$ by $V_M(\mathcal R)$, where $\mathcal R$ is the set of all robots in $C(t)$.
\end{definition}
\subsection{\textbf{Problem Definition}} 
Given a configuration $C(t)$, the goal of min-max gathering is to gather all the robots at one of the min-max nodes. In the initial configuration, all the robots are at the distinct nodes of the grid. If the following conditions are satisfied, we consider a configuration to be \textit{final}.
\begin{itemize}
    \item all the robots are at one of the min-max nodes.
    \item each robot is stationary
    \item during the look phase, any robot taking a snapshot will refrain from moving at time $t$.
\end{itemize}
In other words, the goal of the min-max gathering problem is to transform any initial configuration into a final configuration, and this must be accomplished within a finite amount of time.

\section{Technical Difficulties and Overview of our Strategy} \label{sec:strategy}
In this section, we highlight some key difficulties of the problem in the current model. In the continuous domain, it has been observed that there always exists a unique min-max point, which is the center of the minimum enclosing rectangle of all robots \cite{DBLP:conf/caldam/BhagatM17}. However, things are not as straightforward in the discrete case. In general, the number of min-max nodes can be more than one. In Figure \ref{figure2}, $m_1$ and $m_2$ are the min-max nodes of the configuration. Hence, we have the following observation.
\begin{observation}
   It is possible for a configuration in infinite grids to have multiple min-max nodes.
\end{observation}
\begin{figure}[h]
				\centering
			{
				\includegraphics[width=0.30\columnwidth]{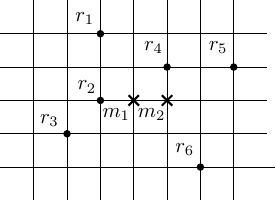}
			}
			\caption{More than one min-max node.}
			\label{figure2}
		\end{figure}
\noindent We next consider the following definition.
\begin{definition}
    The \textit{centrality} of a node $u \in V$ at any time $t$ is defined as $c_t(u)= \sum\limits_{v\in V} d_m(v,u)$ $\mu_t(v)$. A node $u \in V$ is defined as a {\it Weber node} if it minimizes the value $c_t(u)$. The function $\mu_t(v)$ denotes the number of robots at node $v$ at any time $t$. As the robots have weak-multiplicity detection capability, we denote $\mu_t(v)=1$, if $v$ is a robot position and 0 otherwise. In other words, a Weber node $u$ is defined as the node that minimizes the sum of the distances from all the robots to itself.
\end{definition}

  In \cite{DBLP:journals/dc/StefanoN17}, it was proved that a Weber node remains invariant under the movement of a robot towards itself. However, in Figure \ref{figure3}, it may be observed that a min-max node may not remain invariant under the movement of a robot towards itself. In Figure \ref{figure3}, after the movement of the robots $r_1$ and $r_3$ towards $m_1$, $m_1$ does not remain as the min-max node and $m_2$ becomes the unique min-max node in the new configuration. So, we have the following observation.
  \begin{observation}
     If a robot moves towards a min-max node $m$, it may be possible that $m$ is not a min-max node in the new configuration that is formed after the robot moves towards $m$.
  \end{observation}

  \begin{figure}[h]
				\centering
			{
				
    \includegraphics[width=0.60\columnwidth]{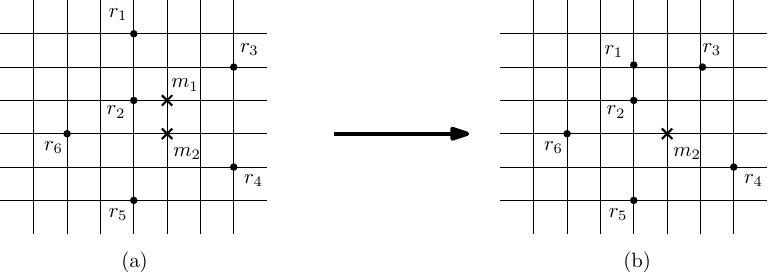}
			}
			\caption{Min-max node is not invariant under the movement of robot towards itself.}
			\label{figure3}
		\end{figure}
  In \cite{DBLP:journals/dc/StefanoN17}, it was also proved that after the movement of a robot towards a Weber node, the set of Weber nodes in the new configuration is a subset of the Weber nodes in the previous case, when the robot was not moving. However, in Figure \ref{figure4}, it may be observed that after the movement of robot $r_1$ towards $m_1$, a new min-max node $m_2$ is created in the new configuration. Next, we have the following observation.
  \begin{observation}
      If a robot moves toward a min-max node $m$, it is possible that after the movement, the new configuration will admit new min-max nodes.
  \end{observation}
   \begin{figure}[h]
				\centering
			{
				
    \includegraphics[width=0.67\columnwidth]{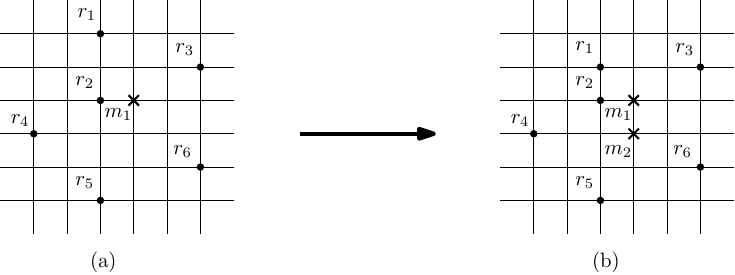}
			}
			\caption{New min-max node is created after the movement of a robot towards a min-max node.}
			\label{figure4}
		\end{figure}
  A common approach to solve any gathering problem is to choose a node that remains invariant under the movement of the robots towards itself. In this way, a multiplicity can be created at one of such nodes and all the other robots can move towards the multiplicity to ensure gathering (robots have multiplicity detection capability). From Observations 1 and 2, it follows that the main challenge with implementing this strategy in our setting is that it is not always true that for any configuration, the set of min-max nodes remains invariant under robot movement. In order to maintain desirable properties, it is necessary to use different strategies and conduct deeper analyses when dealing with min-max nodes.
 
 \subsection{Characterization of the min-max nodes in Infinite Grids}
  First, we consider the following definitions.
  \begin{definition}
     The intersection of all minimal enclosing diamonds (i.e., $\bigcap_{\mathcal{M} \in \mathcal{MED} (\mathcal R)}  \mathcal{M}$) denoted by $\mathcal{IR}(\mathcal R)$ is a rectangle and is called the \textit{Intersection Rectangle} for $\mathcal R$.  
  \end{definition}
Note that the sides of $\mathcal{IR}(\mathcal R)$ are parallel to the sides of the minimal enclosing diamonds. So, we have the following observation.

\begin{observation}
\label{obs:1}
    All robot positions in $\mathcal R$ must be either on the boundaries or strictly inside $\mathcal{IR}(\mathcal R)$. In fact, each boundary
    of $\mathcal{IR}(\mathcal R)$ contains at least one robot position. 
\end{observation}
We consider the following definitions necessary for the characterization of min-max nodes.
\begin{definition}[Sub-graph induced by min-max nodes ($G_M(\mathcal R)$)]
Let $G = (V,$ $ E)$ be the input infinite grid graph. Then $G_M(\mathcal R)$ is the sub-graph of $G$ induced by the set of nodes $V_M(\mathcal R)$.
    \label{def:minmax induced subgraph}
\end{definition}

\begin{definition}[Step-Graph]
    \label{def:Step Graph}
    Let $G=(V,$ $E)$ be the input infinite grid graph. A node-induced sub-graph $G' = (V',$ $E')$ of $G$ is called a step-graph if the following conditions are satisfied.
    \begin{enumerate}
        \item $G'$ is a line graph.
        \item For any three vertices $v_1, v_2$ and $v_3$ in $V'$ where $(v_1,v_2), (v_2,v_3) \in E'$, $v_1, v_2$ and $v_3$ must occupy three corners of a unit square in $G$.  
    \end{enumerate}
\end{definition}
In Figure \ref{step}, the nodes $m_1$, $m_2$ and $m_3$ forms a step-graph.
\begin{figure}
    \centering
    \includegraphics[width=0.34\linewidth]{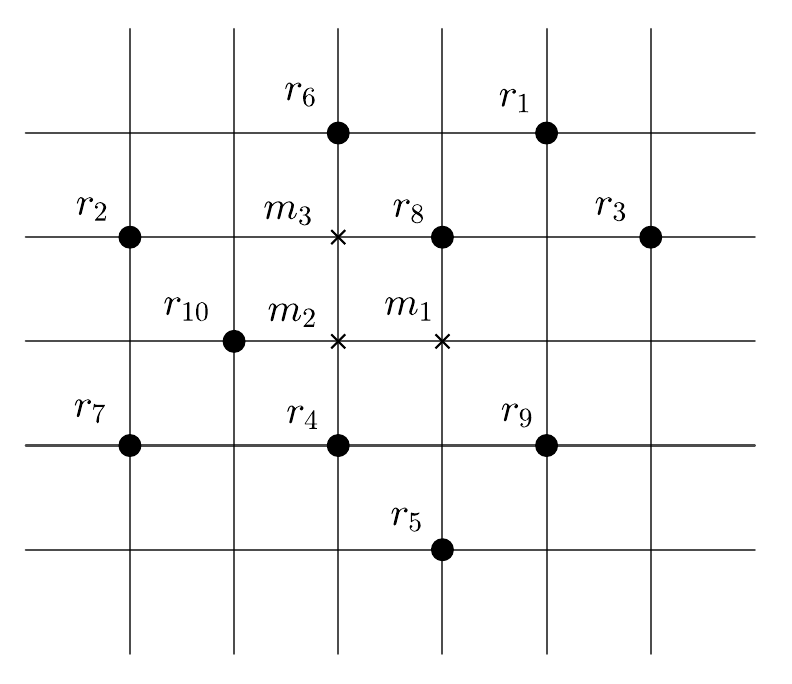}
    \caption{Example of a configuration where $G_M(\mathcal R)$ is a step-graph.}
    \label{step}
\end{figure}
\begin{definition}[Disconnected Step-Graph]
    \label{def: disconnected step graph}
     Let $G=(V,$ $E)$ be the input infinite grid graph. A node-induced sub-graph $G' = (V',$ $E')$ of $G$ is called a Disconnected Step-Graph if the following conditions are satisfied.
     \begin{enumerate}
     \item $|V'| > 1$.
         \item $\forall v \in V'$, $L_V(v)$ and $L_H(v)$ contains no other node of $V'$.
         \item $\forall\  v \in V'$, there is at least one and at most two vertices of $V'$ on $D(v)$.
     \end{enumerate}
\end{definition}
\begin{figure}
    \centering
    \includegraphics[width=0.35\linewidth]{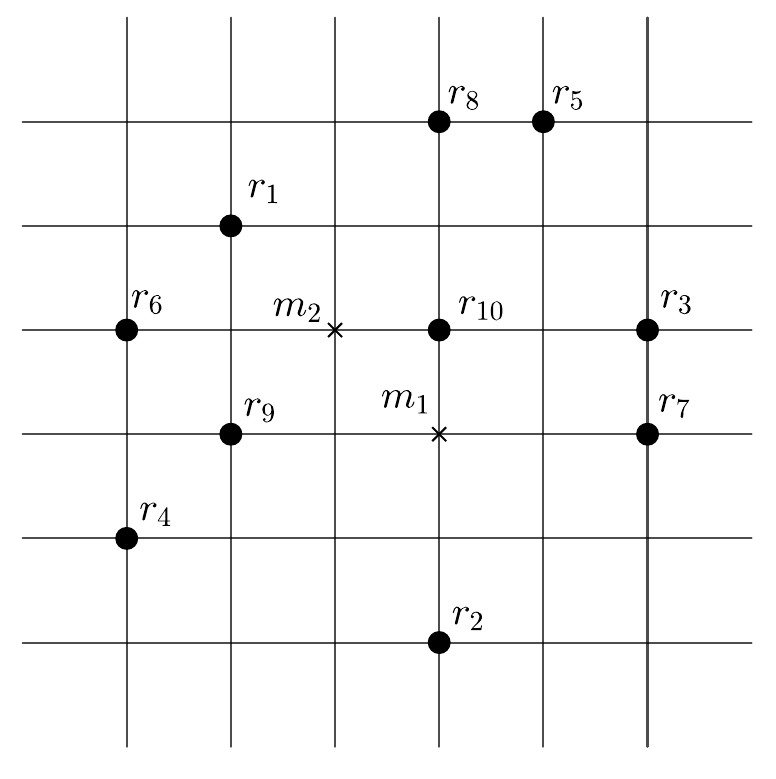}
    \caption{Example configuration, where $G_M(\mathcal R)$ is a disconnected step-graph.}
    \label{disconnected}
\end{figure}
In Figure \ref{disconnected}, the min-max nodes $m_1$ and $m_2$ forms a disconnected step-graph. We next consider the following lemmas.
\begin{lemma}
    \label{lemma: center of a MED(P) is in V_M(P)}
    For a finite set of grid points in $\mathcal R$, if  \Diam $\in$ \MED, then \CNTR $\in V_M(\mathcal R)$.
\end{lemma}
\begin{proof}
    Let, the size of the diamond $\mathcal{M}$ be $S($\Diam$) = k \sqrt{2}$, for all \Diam $\in $ \MED, where $k \in \mathbb{N}$. For the sake of contradiction, assume that there exists some diamond \Diamo $\in $ \MED\  such that $\mathcal{C}($\Diamo$) \notin V_M(\mathcal R)$. Since, \Diamo $\in $ \MED, by the definition of a min-max node there must exist $ u' \in \mathcal R$ such that $d_m(u', \mathcal{C}(\mathcal{M'})) = k = \underset{u\in \mathcal R}{\max}\ d_m(u, \mathcal{C}(\mathcal{M}'))$. As ~$V_M(\mathcal R)~\ne~\phi$ consider $v' \in V_M(\mathcal R)$. This implies, 
    \begin{equation}
    \label{eq:1}
        k = \underset{u\in \mathcal R}{\max}\ d_m(u, \mathcal{C}(\mathcal{M}')) > \underset{u\in \mathcal R}{\max}\ d_m(u, v')
    \end{equation}
Now, assume that $\mathcal{M}_{v'}$ is a diamond with size $S(\mathcal{M}_{v'})=k\sqrt{2}$ and $v' = \mathcal{C}(\mathcal{M}_{v'})$. We claim that $\mathcal{M}_{v'} \in \mathcal R$. In order to prove the claim, it is sufficient to show that for all $u \in \mathcal R$, $d_m(u,v') \le k$. If there is a $u_0 \in \mathcal R$ such that $d_m(u_0,v') > k$, then $\underset{u \in \mathcal R}{\max}\ d_m(u,v') \ge d_m(u_0,v') > k$. This contradicts the assumption that $v' \in V_M(\mathcal R)$. Thus, $\mathcal{M}_{v'} \in  \mathcal {MED (R)}$ and so there must be a robot which is on the boundary of $\mathcal{M}_{v'}$. Thus, $\underset{u \in \mathcal R}{\max}\ d_m(u,v') = k$. Now, from equation~\ref{eq:1}, we get $k > k$, which clearly implies a contradiction. Hence, $\mathcal{C}(\mathcal{M}') \in V_M(\mathcal R)$.
\end{proof}

\begin{lemma}
    \label{lemma: any minmax point is a center of some MED}
    If $v' \in V_M(\mathcal R)$, then $\exists \mathcal{M}_{v'} \in  \mathcal {MED (R)}$ such that $v' = \mathcal{C}(\mathcal{M}_{v'})$.
\end{lemma}
\begin{proof}
     Let, the size of the diamond $\mathcal{M}$ be  $S($\Diam$) = k \sqrt{2}$, for all \Diam $\in \mathcal {MED(R)}$, where $k \in \mathbb{N}$. Now, let $\mathcal{M}_{v'}$ be a diamond such that $S(\mathcal{M}_{v'}) = k \sqrt{2}$ and  $v' = \mathcal{C}(\mathcal{M}_{v'})$, where $v' \in V_M(\mathcal R)$. Now, we claim that $\mathcal{M}_{v'} \in \mathcal {MED(R)}$. Otherwise, if the claim is not true, then there could be two possible cases. 
     
     \noindent Case 1: All the robot positions are strictly inside $\mathcal{M}_{v'}$. This implies that there are no robots on the boundary of $\mathcal{M}_{v'}$. Thus, the size of the minimal enclosing diamond will be of size at most $(k-1)\sqrt{2}$, which contradicts our assumption about the size of minimal enclosing diamonds.

     \noindent Case 2: There exists at least one robot position which is located outside of $\mathcal{M}_{v'}$. If there exists such a robot $u_0$ which is strictly outside of $\mathcal{M}_{v'}$, then $d_m(u_0,v') > k$. This implies, $$\underset{u \in \mathcal R}{\max}\ d_m(u,v') > \underset{u\in \mathcal R\  \&\  \mathcal{M} \in \mathcal{MED}(\mathcal R)}{\max}\ d_m(u,\mathcal{C}(\mathcal{M})) = k.$$ Thus, $v' \notin V_M(\mathcal R)$ which is a contradiction. Hence, $\mathcal{M}_{v'} \in $ \MED\ is the diamond that serves our purpose.
\end{proof}
\begin{lemma}
    \label{lemma: grid line can not contain more than 2 min max point}
    For any grid node $v \in V$, the grid lines $L_H(v)$ and $L_V(v)$ can not contain more than two members of $V_M(\mathcal R)$.
\end{lemma}
\begin{proof}
Let, the size of the diamond $\mathcal{M}$ be  $S($\Diam$) = k \sqrt{2}$, for all \Diam $\in $ \MED, where $k \in \mathbb{N}$.  
    Let there be a grid node $v_0 \in V$ such that at least one of $L_H(v_0)$ and $L_V(v_0)$ contains three members of $V_M(\mathcal R)$. Without loss of generality, let $L= L_H(v_0)$ be the grid line that contains the three min-max nodes $v_1, v_2$ and $v_3$. Let $v_2$ be the central node among these three nodes. Lemma~\ref{lemma: any minmax point is a center of some MED}, asserts that for every min-max node, there exists a minimum enclosing diamond that contains the min-max node as its center. As $v_2$ is a min-max node, there exist $\mathcal{M}_{v_2} \in $ \MED\ such that $v_2 = \mathcal{C}(\mathcal{M}_{v_2})$ and there must exist a robot $u_0$ which is on a boundary of $\mathcal{M}_{v_2}$. Thus, $\underset{u \in \mathcal R}{\max}\ d_m(u, v_2) = d_m(u_0,v_2) = k$. Let $L_\perp$ be the grid line passing through $\mathcal{C}(\mathcal{M}_{v_2})$ that is perpendicular to $L$. Without loss of generality, let $v_1$ be on the left of $L_\perp$ and $v_3$ be on the right of $L_\perp$. Depending on whether $u_0$ is on $L_\perp$ or not on $L_\perp$, the following cases are to be considered.

    \noindent \textit{Case I:} $u_0$ is on $L_\perp$. In this case, $u_0$ must lie on the corners of $\mathcal{M}_{v_2}$ that are not on $L$. Now, for any grid node $v \in V$ that is on the left or right of $v_2$ on $L$, a shortest path from $v$ to $u_0$ contains $v_2$. Thus, $$d_m(v,u_0) = d_m(v,v_2)+d_m(v_2,u_0) = d_m(v,v_2) + k > k,\  [\because v \ne v_2].$$ This implies $v \notin V_M(\mathcal R)$ for all $v \in V$ that are on the left of right of $v_2$ on $L$. Thus, $v_1,v_3 \notin V_M(\mathcal R)$, which is a contradiction. 
    
    \noindent \textit{Case II:} $u_0$ is strictly on the left of $L_\perp$. Then there is a shortest path from $v_3$ to $u_0$ that contains $v_2$. Thus,
    $$d_m(v_3,u_0) = d_m(v_3,v_2)+ d_m(v_2,u_0) \geq 1+k > k.$$ This contradicts the assumption that $v_3 \in V_M(\mathcal R)$. So, $u_0$ can not be on the left of $L_\perp$.

   \noindent  \textit{Case III:} $u_0$ is strictly on the right of $L_\perp$. Then there is a shortest path from $v_1$ to $u_0$ that contains $v_2$. Thus,
    $$d_m(v_1,u_0) = d_m(v_1,v_2)+ d_m(v_2,u_0) \geq 1+k > k.$$ This contradicts the assumption that $v_1 \in V_M(\mathcal R)$. So, $u_0$ can not be on the right of $L_\perp$. 

    \noindent Thus there can not be such robot $u_0$ which is on the boundary of $\mathcal{M}_{v_2}$. Thus, $\mathcal{M}_{v_2} \notin $\MED, which is a contradiction to the assumption that $v_2$ is a min-max node. So, for all $v \in V$, there can not be more than two min-max nodes on $L_H(v)$ and $L_V(v)$.
\end{proof}
\begin{lemma}
    \label{lemma: 2 min max points on same grid line must be neighbours}
    If $v \in V_M(\mathcal R)$, then for any grid line $L \in \{L_H(v),L_V(v)\}$, if $L$ contains another grid node $v' (\ne v)$ such that $v' \in V_M(\mathcal R)$, then $v' \in N(v)$.
\end{lemma}
\begin{proof}
    Let the size of the diamond $\mathcal{M}$ be  $S($\Diam$) = k \sqrt{2}$, for all \Diam $\in $ \MED, where $k \in \mathbb{N}$. Without loss of generality, let $L = L_H(v)$ and there is $v' (\neq v) \in V_M(\mathcal R)$ on $L$ such that $v' \notin N(v)$. Let $v_1$ be the neighbour of $v$ on $L$ between the line segment $\overline{vv'}$. According to Lemma~\ref{lemma: grid line can not contain more than 2 min max point}, note that $v_1 \notin V_M(\mathcal R)$. Thus, there is a robot $u_0$ such that, $d_m(v_1,u_0) > k$. Let $L_\perp$ be the grid line perpendicular to $L$ and passing through $v$. Depending on whether $u_0$ is on $L_\perp$ or not on $L_\perp$, the following cases are to be considered.

    \noindent \textit{Case I:} Let $u_0$ be on $L_\perp$. Then there is a shortest path from $v'$ to $u_0$ that contains $v_1$. Thus, 
    $$d_m(v',u_0) = d_m(v',v_1) + d_m(v_1,u_0) > d_m(v',v_1)+ k > k.$$ This implies, $v' \notin V_M(\mathcal R)$ contrary to the assumption.

\noindent \textit{Case II:} Let $u_0$ be strictly on the left of $L_\perp$. Then there is a shortest path from $v'$ to $u_0$ that contains $v_1$. The rest of the proof follows from the previous case.
    
    \noindent \textit{Case III:} Let $u_0$ be strictly on the right of $L_\perp$. Then there is a shortest path from $v$ to $u_0$ that contains $v_1$. Thus,
    $$d_m(v,u_0) = d_m(v,v_1)+d_m(v_1,u_0)> 1+k >k.$$
    This implies, $v \notin V_M(\mathcal R)$ contrary to the assumption.

    \noindent The above three cases lead to a contradiction of the fact that $v' \notin N(v)$. Thus, if $v' \in L \cap V_M(\mathcal R)$, then $v' \in N(v)$.
\end{proof}

\begin{lemma}
    \label{lemma: two neighbours of a minmax point is minmax condition}
    If there is a \Diam $\in $ \MED\ such that only one $B \in \{\mathcal{B}_{UL}, \mathcal{B}_{DL}, \mathcal{B}_{UR}, \mathcal{B}_{DR}\}$ contains a robot position, then $\exists\  v_1, v_2 \in N(\mathcal{C}(\mathcal{M}))$ such that $v_1, v_2 \in V_M(\mathcal R)$. Furthermore, $v_1, v_2 \in N(\mathcal{C}(\mathcal{M}))$ be such that they are the nearest to $B$ among all nodes in $N(\mathcal{C}(\mathcal{M}))$.
\end{lemma}
\begin{proof}
    Let the size of the diamond $\mathcal{M}$ be $S($\Diam$) = k \sqrt{2}$, for all \Diam $\in $ \MED, where $k \in \mathbb{N}$. According to the statement of the lemma, let \Diam $\in \mathcal{MED}(\mathcal R)$ be such that only one of $\mathcal{B}_{UL}, \mathcal{B}_{DL}, \mathcal{B}_{UR}, \mathcal{B}_{DR}$ contains a robot position. Without loss of generality, assume that the boundary $\mathcal{B}_{UR}$, denoted as $B$, contains a robot position. We also assume that the robot position $u_1$ is on $B$. Let $L= L_H(\mathcal{C}(\mathcal{M}))$ be the horizontal grid line passing through the center and let $L_\perp = L_V(\mathcal{C}(\mathcal{M}))$ be the vertical grid-line passing through the center. Note that $B$ is on the right of $L_\perp$. In this proof, we will first prove the second part of the proof and then proceed to the first part.
    
    \noindent Let $N(\mathcal{C}(\mathcal{M})) = \{v_1,v_2,v_1',v_2'\}$ be such that $v_1$ and $v_2$ are closer to $B$ compared to $v_1'$ and $v_2'$. Note that it is enough to show that $v_1'$ and $v_2'$ are not in $V_M(\mathcal R)$. If $v_1'\in V_M(\mathcal R)$, then there exists a shortest path from $v_1'$ to $u_1$ that contains the node $\mathcal{C}(\mathcal{M})$. This implies that,
    $$d_m(v_1',u_1) = d_m(v_1',\mathcal{C}(\mathcal{M}))+ d_m(\mathcal{C}(\mathcal{M}), u_1) = 1+k > k.$$ which implies that $v_1'$ cannot be a min-max node. The proof holds similarly for $v_2'$. Thus, $v'_1, v'_2 \notin V_M(\mathcal R)$. 

    \noindent Next, we proceed to the first part of the proof, i.e., we have to prove $v_1,v_2 \in V_M(\mathcal R)$. We proof the same for $v_1$. The argument is similar for $v_2$. For the sake of contradiction, let $v_1 \notin V_M(\mathcal R)$. Thus, $\exists\ u_0 \in \mathcal R$ such that $d_m(v_1,u_0) > k$. We claim that $u_0$ can not be on the right of $L_\perp$. Otherwise, there is a shortest path from $\mathcal{C}(\mathcal{M})$ to $u_0$ that contains the node $v_1$, implying that $d_m(\mathcal{C}(\mathcal{M}),u_0) = k+1 > k$. This asserts that $\mathcal{C}(\mathcal{M}) \notin V_M(\mathcal R)$, where $\mathcal{M} \in $ \MED, contradicting Lemma~\ref{lemma: center of a MED(P) is in V_M(P)}. So, $u_0$ is either on $L_\perp$ or on the left of $L_\perp$. Now, among all the boundaries of $\mathcal{M}$, only the boundary $B$ contains robot positions. Let $B_1$ and $B_2$ be the two boundaries of $\mathcal M$ that are on the left of $L_\perp$. Thus, $B_1$ and $B_2$ do not contain any robot positions, implying for all $u \in \mathcal R$ that are on $L_\perp$ or on the left of $L_\perp$, $d_m(u,\mathcal{C}(\mathcal{M}))~\le k-1$, as there exist a shortest path joining $v_1$ and $u$ through $\mathcal C(\mathcal M)$. This again implies that $d_m(v_1,u) = ~d_m(v_1, \mathcal{C}~(\mathcal{M}))+~ d_m(\mathcal{C}(\mathcal{M}), u)~\le 1+~(k-1) =k$. Thus, we get a contradiction as $k < d_m(v_1,u_0) \le k$,  Hence, $v_1 \in V_M(\mathcal R)$.
\end{proof}
\begin{lemma}
    \label{lemma: condition of two neibouring minmax point of a min max point}
     If there is a \Diam $\in$ \MED\ such that only two boundaries $B_1$ and $B_2$ of \Diam\  contain robot positions and the pair $(B_1,B_2) \in $ \{($\mathcal{B}_{UL},\mathcal{B}_{UR}$), ($\mathcal{B}_{UR},\mathcal{B}_{DR}$), ($\mathcal{B}_{DL},\mathcal{B}_{DR}$), ($\mathcal{B}_{DL},\mathcal{B}_{UL}$)\}, then there is exactly one grid node, $v_1 \in N(\mathcal{C}(\mathcal{M}))$ such that $v_1 \in V_M(\mathcal R)$. In fact, $v_1$ is the adjacent node of $\mathcal{C}(\mathcal{M})$ which is nearest to the grid node $B_1 \cap B_2$.
\end{lemma}
\begin{proof}
    Without loss of generality, let for all \Diam $\in$ \MED, $S($\Diam$) = k\sqrt{2}$, where $k \in \mathbb{N}$. Let there be a \Diam $\in$ \MED\ such that only two boundaries $B_1$ and $B_2$ of \Diam\ contain robot positions and the pair $(B_1,B_2) \in $ \{($\mathcal{B}_{UL},\mathcal{B}_{UR}$), ($\mathcal{B}_{UR},\mathcal{B}_{DR}$), ($\mathcal{B}_{DL},\mathcal{B}_{DR}$), ($\mathcal{B}_{DL},\mathcal{B}_{UL}$)\}. Without loss of generality, let $(B_1, B_2) $ be the pair $(\mathcal{B}_{UR}, \mathcal{B}_{UL})$, where $B_1 = \mathcal{B}_{UR}$ and $B_2 = \mathcal{B}_{UL}$ and let $u_1$ be the robot position on $B_1$ and $u_2$ be the robot position on $B_2$. Then $d_m(\mathcal{C}(\mathcal{M}), u_1) = d_m(\mathcal{C}(\mathcal{M}),u_2) = k$.
    
    \noindent We begin the proof by first proving the second part of the proof. Let $v_1 \in N(\mathcal{C}(\mathcal{M}))$ be such that it is nearest to $B_1 \cap B_2$ compared to other three grid nodes in $N(\mathcal{C}(\mathcal{M}))$. For proving the claim, it is sufficient to prove that the other three grid nodes in $N(\mathcal{C}(\mathcal{M}))$ can't be a member of $V_M(\mathcal R)$. Let $v_1' (\ne v_1) \in N(\mathcal{C}(\mathcal{M}))$ such that it is on $L_V(v_1)$. Then $d_m(v_1',u_1) = d_m(v_1',\mathcal{C}(\mathcal{M})) + d_m(\mathcal{C}(\mathcal{M}), u_1) =1+k >k$. Thus, $v_1' \notin V_M(\mathcal R)$. Similarly, let $v_2$ and $v_2'$ be the neighbours of $\mathcal{C}(\mathcal{M})$ that lies on $L_H(\mathcal{C}(\mathcal{M}))$ and are such that $v_2$($\ne v_1$) is on the right and $v_2'(\ne v_1)$ is on the left of $\mathcal{C}(\mathcal{M})$. The proof that $d_m(v_2,u_2) = 1+k > k$ and $d_m(v_2', u_1) =1+k > k$, follows similarly as before. Hence, $v_2, v_2' \notin V_M(\mathcal R)$.

    \noindent Next, we proceed to the first part of the proof. Here, we have to prove that $v_1 \in V_M(\mathcal R)$. If possible, let $v_1 \notin V_M(\mathcal R)$. Thus, there exists at least one such robot position $u_0$ such that $d_m(u_0, v_1) > k$. Now, it may be observed that $u_0$ can not be above the grid-line $L = L_H(\mathcal{C}(\mathcal{M}))$. Otherwise, if the claim holds true, then there must exist a shortest path from $\mathcal{C}(\mathcal{M})$ to $u_0$ that contains $v_1$ and thus, $d_m(\mathcal{C}(\mathcal{M}), u_0) >k+1 > k$. This contradicts the assumption that \Diam $\in $ \MED. So $u_0$ is either on $L$ or lying below $L$. Note that both the boundaries, lying below $L$ (say $B_1'$ and $B_2'$), do not contain any robot positions. Thus, for any robot position $u$ that is on or below $L$, $d_m(\mathcal{C}(\mathcal{M}),u) \le k-1$. Now, for any such $u$ on or below $L$, there is a shortest path from $v_1$ to $u$ that contains $\mathcal{C}(\mathcal{M})$. Thus, for all robot position $u $ on or below $L$, $$d_m(v_1,u) = d_m(v_1,\mathcal{C}(\mathcal{M}))+ d_m(\mathcal{C}(\mathcal{M}), u) \le 1+(k-1) = k.$$ Thus, $k < d_m(v_1,u_0) \le k$, which is a contradiction. Hence, $v_1 \in V_M(\mathcal R)$. 
\end{proof}
\begin{lemma}
    \label{lemma: condition for no neighbour in minmax points}
    If $v \in V_M(\mathcal R)$, the following statements are equivalent.
    \begin{itemize}
        \item [(1)] For all $v' \in N(v)$, $v' \notin V_M(\mathcal R)$.
        \item [(2)] the minimal enclosing diamond $\mathcal{M}_v$ for which $\mathcal{C}(\mathcal{M}_v) =v$ has at least one pair of boundaries $(B_1, B_2) \in \{(\mathcal{B}_{UR},\mathcal{B}_{DL}), (\mathcal{B}_{UL}, \mathcal{B}_{DR})\}$ such that both $B_1$ and $B_2$ contains robot positions.
    \end{itemize}
\end{lemma}
\begin{proof}

   \noindent  \textit{(1) $\implies$ (2)}: Depending on the number of boundaries that contain robot positions, the following cases are to be considered. 
    
    \noindent \textit{Case 1:} Only one boundary contains robot positions. According to Lemma~\ref{lemma: two neighbours of a minmax point is minmax condition}, there must exist two neighbors of $v$ that are min-max nodes. However, according to (1), this is not possible. Therefore, this case is not possible.
    
    \noindent \textit{Case 2:} At least two boundaries contain robot positions.
    
    \noindent \textit{Subcase 1:} Exactly two boundaries contain robot positions (say $B_1$ and $B_2$). According to Lemma~\ref{lemma: condition of two neibouring minmax point of a min max point}, the two boundaries cannot be adjacent. Thus, the pair of two boundaries i.e., $(B_1,$ $B_2)$ must be opposite to each other and hence they must be from either the set $\{(\mathcal{B}_{UR},\mathcal{B}_{DL})$ or from $(\mathcal{B}_{UL}, \mathcal{B}_{DR})\}$.

    \noindent \textit{Subcase 2:} More than two boundaries contain robot positions. In that case, there must exist at least two non-adjacent boundaries that contain robot positions. Thus, the proof follows.

   \noindent  \textit{(2) $\implies $ (1)}: Let $B_1$ and $B_2$ be two opposite boundaries of $\mathcal{M}_v$. Let $u_1, u_2 $ be two robot positions such that $u_1 \in B_1$ and $u_2 \in B_2$. For any $v' \in N(v)$, there is a $u_i$, where $i\in \{1, 2\}$ such that, there is a shortest path from $v'$ to $u_i$ that contains $v$. Thus, $d_m(v',u_i) =k+1 >k$. This implies $v' \notin V_M(\mathcal R), \forall\ v' \in N(v)$. 
\end{proof}
\begin{lemma}
    \label{lemma: condition for unique minmax point}
    If any of the following conditions are true, then $|V_M(\mathcal R)| =1$.
    \begin{itemize}
        \item [\textbf{$(C1)$}] For a diamond \Diam $\in $ \MED, all four boundaries of \Diam\  contain robot positions.
        \item[\textbf{($C2$)}] For a diamond \Diam $\in $ \MED, only three boundaries contain robot positions and there is at least one robot position within distance 1 from the boundary that does not contain robot positions.
        \item[\textbf{($C3$)}] For a diamond \Diam $\in $ \MED, only two opposite boundaries contain robot positions and there is at least one robot position within distance 1 from each of the boundaries that do not contain robot positions.
    \end{itemize}
\end{lemma}

\begin{proof}
Our proof is based on the assumption that one condition is true at a time and then we show that $|V_M(\mathcal R)| = 1$.

   \noindent  \textit{Case I:} Let condition $(C1)$ be true. Thus, there exists a \Diam $\in$ \MED\ such that all the boundaries of \Diam\ contain robot positions. Let the center of the diamond be $v$, i.e., $v = \mathcal{C}(\mathcal{M})$. Then, from Lemma 1, it can be asserted that $v \in V_M(\mathcal R)$. If possible, let us assume that, $|V_M(\mathcal R)| >1$. Thus, there exists a $v' (\ne v) \in V_M(\mathcal R)$. Let $u_1, u_2,u_3,u_4 \in \mathcal R$ such that $u_1 \in \mathcal{B}_{UR},\  u_2 \in \mathcal{B}_{UL},\  u_3 \in \mathcal{B}_{DL}$ and $u_4 \in \mathcal{B}_{DR}$. Then, for any possible position of $v'$ we can always find an $u_i$ where $(i \in \{1, 2, 3 , 4\})$ such that there is a shortest path from $v'$ to $u_i$ that contains the node $v$. As $v$ is a min-max node, $d_m(v',u_i) > k$. This leads to a contradiction. So, $v' \notin V_M(\mathcal R)$ and thus, $|V_M(\mathcal R)| =1$.

  \noindent \textit{Case II:} Let condition $(C2)$ be true. Thus, there exists a diamond \Diam $\in$ \MED such that all boundaries except one (say $B$) have robot positions and there is a robot position $u$ that is within distance 1 from $B$. Without loss of generality, let $B = \mathcal{B}_{DR}$. Clearly, from Lemma 1, $v= \mathcal{C}(\mathcal{M})$ is a min-max node. Now, for the sake of contradiction, let $|V_M(\mathcal R)| >1$. This implies that there is some $v' (\ne v) \in V_M(\mathcal R)$. Note that by Lemma~\ref{lemma: condition for no neighbour in minmax points}, no neighbours of $v$ are min-max nodes. So, $d_m(v,v') \ge 2$. Without loss of generality, let $u_1, u_2,u_3$ be the robot positions such that $u_1 \in \mathcal{B}_{UR},\ u_2\in \mathcal{B}_{UL}$ and $u_3 \in \mathcal{B}_{DL}$. Let $v'$ be on the left of $L_V(v)$ and above $L_H(v)$. Then, there is a shortest path from $v'$ to $u$ that contains $v$ which implies, $d_m(v',u) = d_m(v',v)+d_m(v,u) \ge 2+(k-1)=k+1 >k$. For all other possible positions of $v'$ we can always find an $u_i$ where $i \in \{1, 2, 3\}$ such that a shortest path from $v'$ to $u_i$ contains $v$ and thus $d_m(v',u_i) > k$. This implies $v' \notin V_M(\mathcal R)$ and thus $|V_M(\mathcal R)| =1$.

    \noindent \textit{Case III:} Let the condition $(C3)$ be true. The proof of this case proceeds similarly as before, as in \textit{Case II}.
\end{proof}

\begin{lemma}
    \label{sufficiency for condition of one neighbour}
    Let $v \in V_M(\mathcal R)$ has exactly one neighbour $v_1$ such that, $v_1 \in V_M(\mathcal R)$. Then $\mathcal{M}_V  \in$ \MED, for which $\mathcal{C}(\mathcal{M}_v) = v$ has exactly two adjacent boundaries $B_1$ and $B_2$ that contain robot positions. In fact, $v_1$ is nearest to $B_1 \cap B_2$ than any other neighbours of $v$.  
\end{lemma}
\begin{proof}
    First note that if there exist robot positions on all the four boundaries of $\mathcal{M}_V$, then according to Lemma~\ref{lemma: condition for no neighbour in minmax points}, there is exactly one min-max node, which is not possible according to the statement of the lemma. Next, if there exist three robot positions on the boundary of $\mathcal{M}_V$, it may be observed that depending on the position of the boundary nodes, both $v$ and $v_1$ cannot be min-max nodes. According to Lemma~\ref{lemma: two neighbours of a minmax point is minmax condition}, it can also be observed that $\mathcal{M}_v$ can not have robot positions in exactly one of its boundaries. Thus, $\mathcal{M}_v$ has exactly two boundaries, (say $B_1$ and $B_2$), that contain robot positions. Now by Lemma~\ref{lemma: condition for no neighbour in minmax points}, $B_1$ and $B_2$ can not be opposite. Hence, $B_1$ and $B_2$ must be two adjacent boundaries of $\mathcal{M}_v$ that contains robot positions. In fact, $v_1$ is nearest to $B_1 \cap B_2$, according to Lemma~\ref{lemma: condition of two neibouring minmax point of a min max point}.
\end{proof}
\begin{lemma}
    \label{sufficiency for condition of two neighbour }
    Let $v \in V_M(\mathcal R)$ has exactly two neighbours $v_1$ and $v_2$ such that, $v_1,v_2 \in V_M(\mathcal R)$. Then $\mathcal{M}_V  \in$ \MED, for which $\mathcal{C}(\mathcal{M}_v) = v$ has exactly one boundary $B_1$ that contains robot positions. In fact, $v_1$ and $v_2$ are the nearest to $B_1$ than any other neighbours of $v$.  
\end{lemma}
\begin{proof}
   Assume that there exists a node $v \in V_M(\mathcal R)$ that has two neighbours $v_1$ and $v_2$, which are also min-max nodes. Now, according to Lemma~\ref{lemma: grid line can not contain more than 2 min max point}, $v_1,v_2$ and $v$ can not be on the same grid line. According to Lemma~\ref{lemma: condition of two neibouring minmax point of a min max point} and Lemma~\ref{lemma: condition for no neighbour in minmax points} and similar to before as in Lemma~\ref{sufficiency for condition of one neighbour}, $\mathcal{M}_v$ can not have more than two boundaries that contain robot positions. This implies $\mathcal{M}_v$ has exactly one boundary, say $B_1$, that contains robot positions. Now, from Lemma~\ref{lemma: two neighbours of a minmax point is minmax condition}, $v_1$ and $v_2$ are the neighbours of $v$ that are nearest to $B_1$, compared to all the other neighbours of $v$.
\end{proof}
\begin{lemma}
\label{lemma: condition diagonal}
    Let $v \in V_M(\mathcal R)$ and $|V_M(\mathcal R)| > 1$. Also let for all $v' \in N(v)$, $v' \notin V_M(\mathcal R)$. Then, $\exists\ v_0 \in D(v)$ such that $v_0 \in V_M(\mathcal R)$. Furthermore, $\forall\  v_0' \in N(v_0)$, $v_0' \notin V_M(\mathcal R)$.
\end{lemma}
\begin{proof}
    Let $v \in V_M(\mathcal R)$ and let $\mathcal{M}_v \in $\MED\ be such that $\mathcal{C}(\mathcal{M}_v) = v$ and let $S(\mathcal{M}_v) =k \sqrt{2}$, where $k \in \mathbb{N}$. According to Lemma~\ref{lemma: condition for no neighbour in minmax points}, since no neighbour of $v$ are min-max nodes, there must exist at least a pair of opposite boundaries $B_1$ and $B_2$ of $\mathcal{M}_v$ (say) that contains robot positions. Without loss of generality, let $B_1 = \mathcal{B}_{UR}$ and $B_2 = \mathcal{B}_{DL}$. According to Lemma~\ref{lemma: condition for unique minmax point}, since $|V_M(\mathcal R)| >1$, at least one of the remaining two boundaries must not contain robot positions. Without loss of generality, let $B = \mathcal{B}_{DR}$ be such a boundary. Also, by Lemma~\ref{lemma: condition for unique minmax point}, there cannot be any robot position within distance 1 from $B$. Now, we claim that the node $v_0 \in D(v)$ that is furthest from $B$ out of the four diagonal nodes is a min-max node. Otherwise, if $v_0 \notin V_M(\mathcal R)$, there exists a robot position $\ u_0$ such that $d_m(v_0,u_0) > k$. We first show that $u_0$ can't be above $L_H(v)$ or, on the left of $L_V(v)$. Let $v_0'$ be the node $ L_H(v_0) \cap L_V(v)$ and $v_0''$ be the node $L_H(v) \cap L_V(v_0)$. Note that $d_m(v, v_0') = d_m(v,v_0'') = d_m(v_0,v_0') = d_m(v_0,v_0'') =1$. Now, for any robot position $u$ which is above $L_H(v)$, there is a shortest path from $v$ to $u_0$ that contains $v_0'$. Thus, $d_m(v, u_0) = d_m(v,v_0')+ d_m(v_0',u_0) = d_m(v_0,v_0')+d_m(v_0',u_0) = d_m(v_0,u_0) >k$. This is a contradiction to our assumption. In a similar way, $u_0$ can't be on the left of $L_V(v)$ if node $v_0''$ is considered. Thus, $u_0$ must be in the region $\mathcal{RG}$, which contains all grid nodes to the right of $v$ on $L_H(v)$, all grid nodes below $v$ on $L_V(v)$, and all grid nodes to the right $L_V(v)$ and below $L_H(v)$. Note that for all robot positions $u$ that are in $\mathcal{RG}$, $d_m(v,u) \le k-2$ [as $B$ is the only boundary of $\mathcal{M}_v$ in this region and all nodes that are on $B$ and at a distance 1 from $B$ does not contain any robot positions]. Now for all robot positions $u \in \mathcal{RG}$, there is a shortest path from $v_0$ to $u$ that contains $v$. So, $d_m(v_0,u) = d_m(v_0,v)+d_m(v,u) \le 2+k-2 =k$ ( $\forall u \in \mathcal{RG}$ and $u$ is a robot position), which is a contradiction to the assumption that $v_0 \notin V_M(\mathcal R)$. Hence, $v_0 \in V_M(\mathcal R)$.

    \noindent Now, in $\mathcal{M}_v$, $B_1 = \mathcal{B}_{UR}$ and $B_2 = \mathcal{B}_{DL}$ contains robot positions and $B = \mathcal{B}_{DR}$ does not contain any robot positions. Note that since $B = \mathcal{B}_{DR}$ does not contain any robot positions, the boundaries $\mathcal{B}_{UR}$ and $\mathcal{B}_{DL}$ of $\mathcal{M}_{v_0}$ must contain robot positions, where $\mathcal{M}_{v_0} \in $ \MED\ with $\mathcal{C}(\mathcal{M}_{v_0}) = v_0$. Thus by Lemma~\ref{lemma: condition for no neighbour in minmax points}, $v_0' \notin V_M(\mathcal R), \forall\ v_0' \in N(v_0)$.
\end{proof}

\begin{lemma}
    \label{lemma: 1 component tree implies step graph}
    If a component of $G_M(\mathcal R)$ is a tree then it must be a step-graph.
\end{lemma}
\begin{proof}
    Let $C$ be a component of $G_M(\mathcal R)$ that is a tree. Therefore, it must be connected. First, assume that $C$ contains a single node $v_1$. In that case, it must be a step-graph. If $C$ contains exactly two nodes, then the nodes must be adjacent and hence, $C$ must be a step-graph. Next, assume the case when $C$ contains at least three nodes. This implies that only two cases are possible: $(i)$ there exist three nodes $v_1$, $v_2$ and $v_3$ such that $v_1$ and $v_2$ are adjacent in $C$, $v_2$ and $v_3$ are adjacent in $C$ and occupying three corners of a unit square and $(ii)$ there exist three nodes $v_1$, $v_2$ and $v_3$ such that $v_1$ and $v_2$ are adjacent in $C$, $v_2$ and $v_3$ are adjacent in $C$ and are positioned on the same grid-line. Note that according to Lemma \ref{lemma: grid line can not contain more than 2 min max point}, case $(ii)$ is not possible. Thus, the only case remain is the case $(i)$. Therefore, $C$ must be a step-graph, and hence the lemma is proved.
    
\end{proof}
\begin{lemma}
    If a component of $G_M(\mathcal R)$ contains a cycle then the component must be a cycle of length four.
    \label{lemma: if G_M(P) contains a cycle then G_M(P) is a cycle of length 4}
\end{lemma}
\begin{proof}
  Let $C$ be a component of $G_M(\mathcal R)$ which contains a cycle. Depending on whether $C$ is a cycle or not, the following two cases may arise.
  
  \paragraph{Case-1} $C$ is not a cyclic graph but has a proper subgraph $C'$, which is a cycle. We have to prove that $C'$ is a cycle of length four. Otherwise, if not, then the cycle must be of length greater than four. Note that the length of the cycle must be even and must be at least six. In this case, there are at least three grid-lines that include min-max nodes. Additionally, at least two of these grid-lines will contain three or more min-max nodes. This is a contradiction to the assumption, as according to Lemma \ref{lemma: grid line can not contain more than 2 min max point}, three min-max nodes can not be on any grid-line. Therefore, $C'$ must be a cycle of length four.

  \paragraph{Case-2} Suppose the component $C$ is a cycle. We need to prove that the cycle has a length of exactly four. The argument follows the same reasoning as in Case 1, since the cycle's length cannot be greater than six.

 \noindent So, the component $C$ of $G_M(\mathcal R)$ must be a cycle of length exactly four.
  \end{proof}
  
  \begin{corollary}
     \label{cor: connected G_M(P) is either a step graph or a cycle of length 4}
     If $G_M(\mathcal R)$ is connected, then it is either a step-graph or a cycle of length four.
 \end{corollary}
 \begin{proof}
    The proof follows from Lemma~\ref{lemma: 1 component tree implies step graph} and Lemma~\ref{lemma: if G_M(P) contains a cycle then G_M(P) is a cycle of length 4}. 
 \end{proof}
   \begin{lemma}
    \label{lemma: one component or totally disconnected}
If $|V_M(\mathcal R)| = m$ and $m>1$, then the graph $G_M(\mathcal R)$ is either connected or, has exactly $m$ components.
\end{lemma}
\begin{proof}
We will prove the result by contradiction. Assume by contradiction that $G_M(\mathcal R)$ has more than one but less than $m$ components. Thus, according to the Pigeonhole principle, there must exist a component $C_1$ which has at least two nodes of $V_M(\mathcal R)$. Let $\mathcal{SER}(C_1)$ be the smallest enclosing rectangle for the nodes in $C_1$. Without loss of generality, assume that the grid-lines which are also the boundaries of $\mathcal{SER}(C_1)$ be $L_l$, $L_r$, $L_u$ and $L_d$. Let $L_l$ be the boundary on the left, $L_r$ be the boundary on the right, $L_u$ be the upper boundary and $L_d$ be the lower boundary. Let $C$ be another component of $G_M(\mathcal R)$. It is important to note that all grid-lines between $L_u$ and $L_d$ (and similarly, between $L_l$ and $L_r$), including the nodes in those lines, cannot contain any node in $C$. Otherwise, there will be a grid-line $L$ that satisfies any one of the following conditions:
\begin{enumerate}
    \item $L$ contains three nodes of $G_M(\mathcal R)$.
    \item $L$ contains two nodes of $G_M(\mathcal R)$ but they are not adjacent.
\end{enumerate}
Note that condition 1 violates Lemma~\ref{lemma: grid line can not contain more than 2 min max point} and condition 2 violates Lemma~\ref{lemma: 2 min max points on same grid line must be neighbours}. Consequently, the locations of nodes in $C$ may be any of the following.
\begin{itemize}
    \item Above $L_u$ and on the right of $L_r$.
    \item Above $L_u$ and on the left of $L_l$.
    \item Below $L_d$ and on the right of $L_r$.
    \item Below $L_d$ and on the left of $L_l$.
\end{itemize}
 Without loss of generality, let us assume $C$ is below $L_d$ and on the right side of $L_r$. Let $v$ be a node in $C_1$ such that $v$ is the rightmost node of $C_1$ on $L_d$. Let $\mathcal{M}_v$ be the minimal enclosing diamond for which $\mathcal{C}(\mathcal{M}_v) = v$, with size of $\mathcal{M} = k\sqrt{2}$, for all $\mathcal{M} \in$ \MED. We claim that the boundary $\mathcal{B}_{UL}$ of $\mathcal{M}_v$ always contain robot positions. For that, first note that $deg(v) \le 2$, where $deg (v)$ denotes the degree of $v$ in $C_1$. Depending on the degree of $v$, there can be two cases.
 
 \noindent \textit{Case 1:} $deg(v) =2$. Assume that without loss of generality, $v_1$ and $v_2$ are the neighbors of $v$ that are in $ C_1$. It should be noted that $v_1$ and $v_2$ are the vertices lying directly above $v$ on $L_V(v)$ and directly on the left of $v$ on $L_H(v)$. Thus, for this case $\mathcal{B}_{UL}$ of $\mathcal{M}_v$ is nearest to $v_1$ and $v_2$. So, according to Lemma~\ref{sufficiency for condition of two neighbour }, only $\mathcal{B}_{UL}$ of $\mathcal{M}_v$ must contain a robot position among all the boundaries of $\mathcal{M}_v$. 
 
 \noindent \textit{Case 2:} $deg(v) =1$. Let $v_1$ be the neighbor of $v$ that is in $C_1$. It should be noted that $v_1$ must be lying either directly above $v$ on $L_V(v)$, or directly on the left of $v$ on $L_H(v)$. Also, if $v_1$ is on the left of $v$ on $L_H(v)$, then only for the adjacent boundaries $\mathcal{B}_{UL}$ and $\mathcal{B}_{DL}$, $v_1$ is nearest to $\mathcal{B}_{UL} \cap \mathcal{B}_{DL}$. Thus, according to Lemma~\ref{sufficiency for condition of one neighbour}, exactly the boundaries $\mathcal{B}_{UL}$ and $\mathcal{B}_{DL}$ contain robot positions. On the other-hand, if $v_1$ is above $v$ on $L_V(v)$, then only for the adjacent boundaries $\mathcal{B}_{UL}$ and $\mathcal{B}_{UR}$, $v_1$ is nearest to $\mathcal{B}_{UL} \cap \mathcal{B}_{UR}$. Thus, again according to Lemma~\ref{sufficiency for condition of one neighbour}, exactly the boundaries $\mathcal{B}_{UL}$ and $\mathcal{B}_{UR}$ contain robot positions. 
 
 \noindent Thus, for all the above cases $\mathcal{B}_{UL}$ contains a robot positions. Let $u_0$ be such a robot position and $v_0$ be a node in $C$. Then, there is a shortest path from $v_0$ to $u_0$ that contains $v$. Thus, $d_m(v_0,u_0) = d_m(v_0, v) + d_m(v, u_0)= d_m(v_0,v)+k \ge k+2 > k$ [as $v$ and $v_0$ are in different components $d_m(v_0, v) \ge 2$]. Thus $v_0 \notin V_M(\mathcal R) \implies v_0 \notin G_M(\mathcal R)$. This is a contradiction to our assumption. Thus $G_M(\mathcal R)$ either has only one component (i.e., $G_M(\mathcal R)$ is connected) or $G_M(\mathcal R)$ has exactly $m$ components.
\end{proof}

\begin{lemma}
    \label{lemma: disconnected G_M(P) is a disconnected step graph} If $G_M(\mathcal R)$ is disconnected, then $G_M(\mathcal R)$ is a Disconnected Step-Graph.
\end{lemma}
\begin{proof}
    Let $G_M(\mathcal R) = (V',E')$ be a disconnected graph, where $V' = V_M(\mathcal R)$. Thus, $|V'| > 1$. By Lemma~\ref{lemma: one component or totally disconnected}, each component of $G_M(\mathcal R)$ has exactly one node of $V'$. This implies for all $v \in V'$, if $v' \in N(v)$, then $v' \notin V'$. To prove that $G_M(\mathcal R)$ is a disconnected step-graph, we have to prove three statements according to Definition 14. Firstly, it has been already observed that $|V'| > 1$. Next, we will prove that $\forall v \in V'$, $L_V(v)$ and $L_H(v)$ contains no other node of $V'$. If it contains a node of $V'$, there exists a grid-line that contains two vertices of $V'$ and the two vertices are not neighbors of each other. This contradicts Lemma~\ref{lemma: 2 min max points on same grid line must be neighbours} and hence $\forall v \in V'$, $L_V(v)$ and $L_H(v)$ contains no other node of $V'$. Next, we will prove that $\forall\  v \in V'$, there is at least one and at most two vertices of $V'$ on $D(v)$. If there exist more than two vertices on $D(v)$, then there must exist a grid-line that contains two vertices of $V'$ and they are not neighbors of each other, contradicting Lemma~\ref{lemma: 2 min max points on same grid line must be neighbours}. Otherwise, if there exists no neighbors of $v$ that are in $V'$, this will contradict Lemma~\ref{lemma: condition diagonal}. Hence, if $G_M(\mathcal R)$ is disconnected, then it must be a disconnected step-graph.
\end{proof}
From Corollary~\ref{cor: connected G_M(P) is either a step graph or a cycle of length 4} and Lemma~\ref{lemma: disconnected G_M(P) is a disconnected step graph}, we can state the following theorem.
\begin{theorem}
\label{Thm: Characterization of $G_M(P)$}
The graph $G_M(\mathcal R)$ is any one of the following graphs:
\begin{enumerate}
    \item Step-Graph.
    \item Disconnected Step-Graph.
    \item Cycle of length four.
\end{enumerate}
\end{theorem}
In Figure \ref{Fourcycle}, the min-max nodes $m_1$, $m_2$, $m_3$ and $m_4$ forms a four-cycle graph.
\begin{figure}
    \centering
    \includegraphics[width=0.27\linewidth]{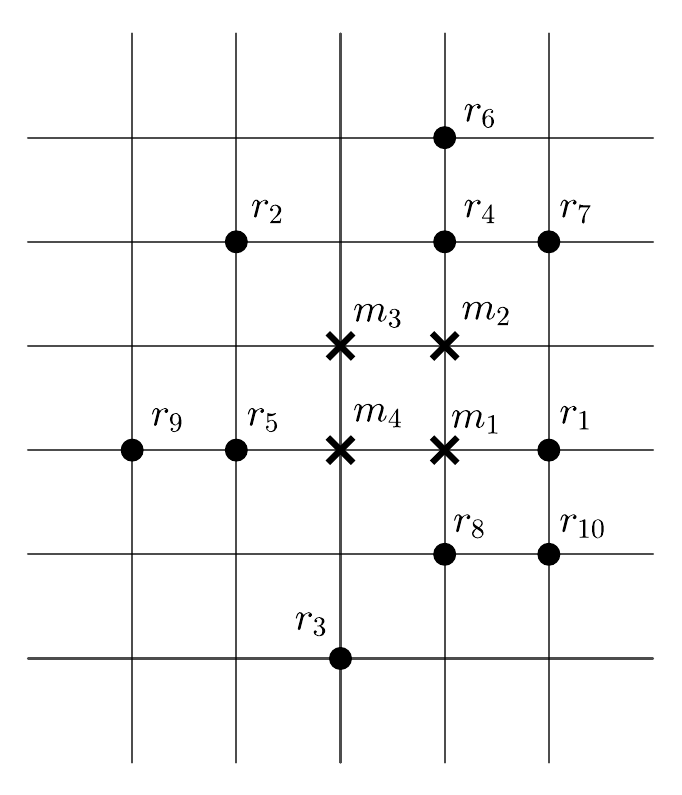}
    \caption{Example of a configuration where $G_M(\mathcal R)$ is a four-cycle graph.}
    \label{Fourcycle}
\end{figure}
\subsection{Impossibility Results}
In this subsection, we discuss all those initial configurations for which the min-max gathering is impossible. First, we consider the definition and an impossibility result related to partitive automorphisms \cite{StefanoN13}.

\begin{definition}[Partitive Automorphism]
Given an automorphism $\phi$ $\in$ $Aut(C(t))$, the cyclic subgroup of order $k$ generated by $\phi$ is given by $\lbrace\phi^{0}, \phi^{1}= \phi, \phi^2= \phi\circ\phi, \ldots, \phi^{k-1} \rbrace$, where $\phi^{0}$ denotes the identity of the cyclic subgroup. We define a relation $\rho$ on the set of nodes $V$, stating that two nodes $x$ and $y$ are related if there exists an automorphism $\gamma$ in some subgroup $H$ of $Aut(C(t))$ such that $\gamma(x)=y$. Note that the relation $\rho$ is an equivalence relation defined on the set of nodes $V$. Due to the equivalence relation defined on the set of nodes $V$, $V$ is partitioned into disjoint subsets called equivalence classes or \textit{orbits}. We denote the orbit of $x$ by $H(x)$.  
\end{definition}
\begin{theorem}[\cite{StefanoN13}]
\label{thm: impossiblePartitive}
    If a configuration admits a partitive automorphism, then it is ungatherable.
\end{theorem}
\noindent Assume that the configuration is symmetric with respect to a single line of symmetry $l$. If there exists a robot position on $l$, the configuration can be transformed into an asymmetric configuration by moving a robot on $l$ towards one of its adjacent nodes, away from $l$. Similarly, if there exists a min-max node on $l$, then the gathering is ensured at one of the min-max nodes on $l$. Therefore, consider the case when $l$ contains no robot positions or min-max nodes. Note that, if the configuration is such that there exist no robots or min-max nodes on $l$, then $l$ must pass through the center of the edges, and hence the configuration must be partitive. Hence, we have the following lemma.
\begin{lemma} \label{lemma 16}
    If the configuration is symmetric with respect to a single line of symmetry $l$ and $l \cap (\mathcal R \cup V_M(\mathcal R))= \phi$, then the configuration is ungathgerable. 
\end{lemma}
Next, assume the case when the configuration admits a rotational symmetry with the center of rotation $c$. If $c$ does not contain any robot or min-max node, then note that in that case, $c$ must be either a center of an edge or the center of an unit square. In other words, the configuration must be partitive. Hence, we have the following lemma.
\begin{lemma} \label{lemma17}
    If the configuration is symmetric with respect to rotational symmetry and $\lbrace c \rbrace \cap (\mathcal R \cup V_M(\mathcal R))= \phi$, then the configuration is ungatherable.
\end{lemma}
Next, we will prove some lemmas that are necessary to characterize the configurations for which gathering is impossible.
\begin{lemma} \label{lemma 18}
   If a configuration $C(t)$ admits either a vertical or horizontal axis of symmetry and $G_M(\mathcal{R})$ is a four cycle, then $C(t)$ must admit a partitive automorphism.
\end{lemma}
\begin{proof}
Let $ABCD$ be the intersection rectangle $\mathcal{IR}(\mathcal{R})$, for the configuration $C(t)$. Without loss of generality, let $C(t)$ admits a vertical line of symmetry, i.e., $C(t)$ admits an isomorphism, $\phi_{reflect} \in Aut(C(t))$, such that $<\phi_{reflect}>$= \{$e$, $\phi_{reflect}\}$. Let us consider the vertical line $AC$. We claim that $AC$ is the line of symmetry for the configuration. Otherwise, without loss of generality, let $L$ be the vertical axis of symmetry, which is on the left of $AC$. In that case, let $r$ be a robot on the side $AB$. Note that according to Observation~\ref{obs:1}, such robots always exist on the boundary of $\mathcal{IR}(\mathcal{R})$. Hence, there must exist a robot $r'$ such that, $r'(t)$ is the reflectional image of $r(t)$ about $L$. Thus, $r'(t)$ must be outside of $\mathcal{IR}(\mathcal{R})$ in $C(t)$, a contradiction of Observation~\ref{obs:1}. Similarly, $L$ cannot be an axis of symmetry if it is strictly on the right of $AC$. Thus, $AC$ is the only axis of symmetry. Now, since $AC$ doesn't pass through any grid nodes, the automorphism generated by the $L$ must be a partitive automorphism. Hence, the lemma is proved. The proof is similar when the configuration $C(t)$ admits a horizontal line of symmetry.

\end{proof}
\begin{figure}[h]
    \centering
    \includegraphics[scale=0.2]{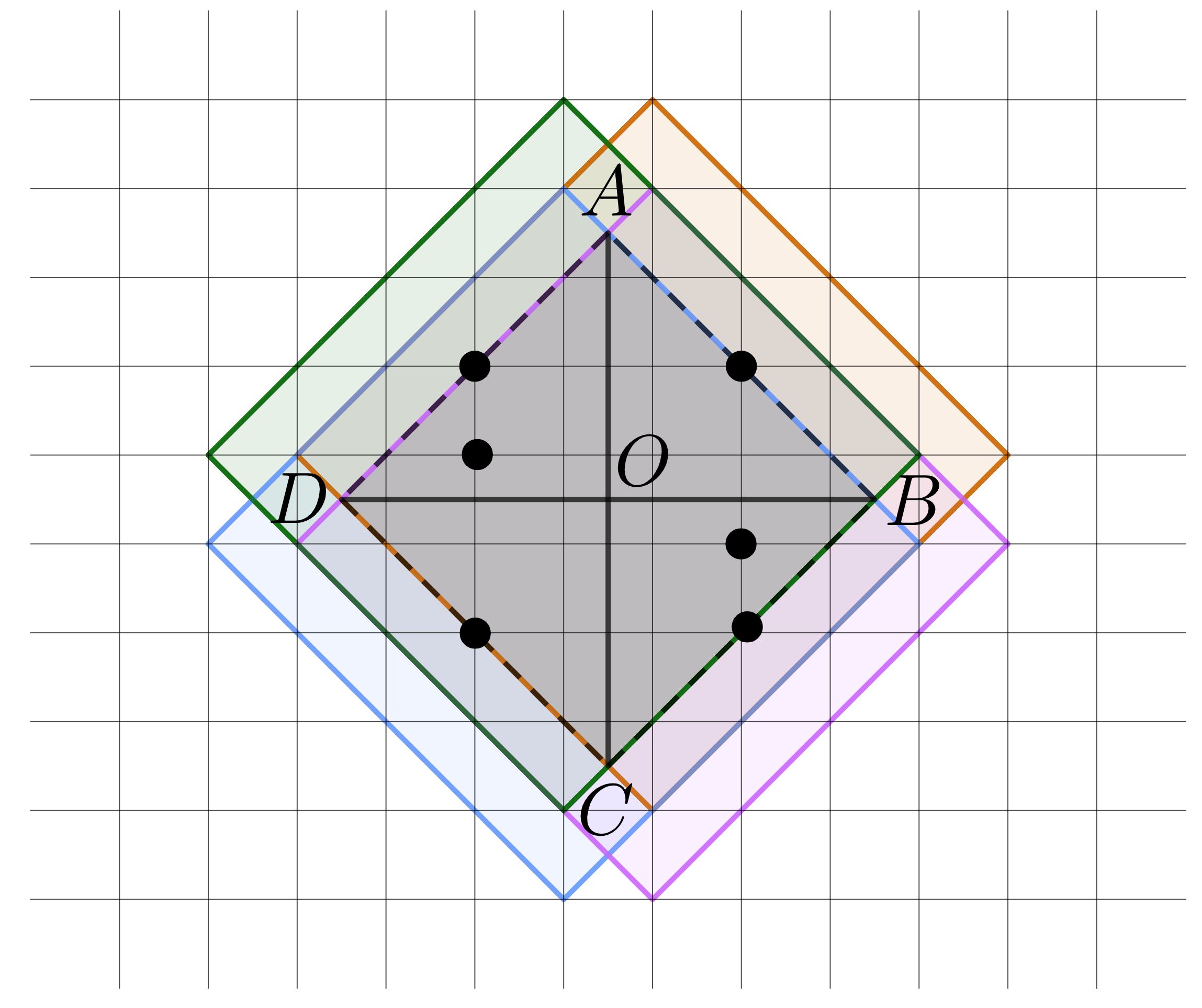}
    \caption{Rotational symmetric configuration with center of rotation $O$, where $V_M(\mathcal{R})$ is a four cycle. }
    \label{fig:enter-label}
\end{figure}
\begin{lemma} \label{impossible3}
   If a configuration $C(t)$ admits a rotational symmetry and if $G_M(\mathcal{R})$ is a four-cycle, then $C(t)$ must admit a partitive automorphism.    
\end{lemma}
\begin{proof}
    Let $ABCD$ be the intersection rectangle $\mathcal{IR}(\mathcal{R})$, for the configuration $C(t)$. Without loss of generality, let $C(t)$ admit a rotational symmetry about a point of rotation (i.e., $C(t)$ admits an isomorphism, $\phi_{rot} \in Aut(C(t))$, such that $|<\phi_{rot}>|$= 2 or, 4). Now, it may be observed that the center of rotation must be the point of intersection of the two diagonals $AC$ and $BD$ (Figure \ref{fig:enter-label}). Note that none of the lines $AC$ and $BD$ passes through any grid nodes, implying their point of intersection is not a grid node. This again implies that the center of rotation is either the mid-point of an edge or the center of a unit square. Hence, $\phi_{rot}$ must be a partitive automorphism.
\end{proof}
Now, as a result of Theorem~\ref{thm: impossiblePartitive}, we have the following corollary.
\begin{corollary} \label{corollary2}
  If the initial configuration $C(0)$ is such that $V_M(\mathcal{R})$ is a four-cycle, then min-max gathering is unsolvable if $C(0)$ is any one of the following:
  \begin{enumerate}
      \item $C(0)$ admits reflectional symmetry about a vertical or horizontal axis of symmetry.
      \item $C(0)$ admits rotational symmetry about a node.
  \end{enumerate}
\end{corollary}
As a result, when the initial configuration is such that $V_M(\mathcal R)$ is a four-cycle, the configuration must be either asymmetric or must admit a diagonal line of symmetry. Next, consider the case when $V_M(\mathcal R)$ is a disconnected step-graph. We have the following lemmas that explain the ungatherable configurations for this case.
\begin{lemma}\label{lemma 20}
    Let $C(t)$ be a configuration that is symmetric with respect to either a vertical or horizontal axis of symmetry and $G_M(\mathcal{R})$ is a disconnected step-graph, then $C(t)$ admits a partitive automorphism.
\end{lemma}
\begin{proof}
Assume that the graph $G_M(\mathcal{R})$ is a disconnected step-graph and let $G_M(\mathcal{R})$ has $p$ nodes, where $p \geq1$. Let $ABCD$ be the intersecting rectangle $\mathcal{IR}(\mathcal{R})$. Observe that the sides $AD$ and $BC$ are of length $k\sqrt{2}$ and sides $AB$ and $CD$ are of length $(k-p+1)\sqrt{2}$, where $k\sqrt{2}$ is the side of the minimal enclosing diamond. Thus, for $p > 1$, $\mathcal{IR}(\mathcal{R})$ is not a square. Let us order the minimal enclosing diamonds first. We denote a minimal enclosing diamond as $\mathcal{M}_j$ if its center is the $j$-th topmost min-max node in $C(t)$. It should be noted that $BC= \mathcal{B}_{DR}$ of $\mathcal{M}_1, $ and $DA= \mathcal{B}_{UL}$ of $\mathcal{M}_p$. Also, note that $\mathcal{B}_{UL}$ of $\mathcal{M}_1$ and $\mathcal{B}_{DR}$ of $\mathcal{M}_p$ and both of $\mathcal{B}_{DR}$ and $\mathcal{B}_{UL}$ of $\mathcal{M}_j$ does not intersect with $\mathcal{IR}(\mathcal{R})$, where $j \ne 1, p$.

Let $L_1$ be a vertical diagonal of $\mathcal{M}_p$ and $L_2$ be the vertical diagonal of $\mathcal{M}_{p-1}$. Assume that $C(t)$ admits a vertical axis of symmetry, say $L$. We claim that $L$ can not be $L_1$ or any vertical line on the right of $L_1$. Similarly, $L$ can not be $L_2$ or any vertical line on the left of $L_2$. First, we prove that $L$ can not be $L_1$ or any vertical line on the right of $L_1$. Let $r_1(t)$ be a robot position on the side $CD$ (say). Note that $CD$ lies strictly to the left of $L_1$ due to the fact that $ABCD$ is a non-square rectangle. If $L$ is the line $L_1$ or any grid line on the right of $L_1$, the symmetric image of $r_1(t)$, $r_1'(t)$ (say) must be either on $\mathcal{B}_{DR}$ of $\mathcal{M}_p$ (as $CD$ is a segment of $\mathcal{B}_{DL}$ of $\mathcal{M}_p$ and $L_1$ is the vertical diagonal of $\mathcal{M}_p$) or it is outside of the region $\underset{i\in [1,p] \cap \mathbb{N}}{\bigcup} \mathcal{M}_i$. In both case $r_1'(t) \notin \mathcal{IR}(\mathcal{R})$ which is a contradiction to Observation~\ref{obs:1}. So $L$ must be a line lying strictly on the left of $L_1$.

Next, assume that either $L=L_2$ or $L$ is on the left of $L_2$. Consider any robot position on $AB$, say $r_2(t)$. Since $ABCD$ is a non-square rectangle, $AB$ lies strictly to the right of $L_2$. Also, since $AB$ is a segment of $\mathcal{B}_{UR}$ of $\mathcal{M}_{p-1}$ and $L_2$ is the vertical diagonal of $\mathcal{M}_{p-1}$, the image of $r_2(t)$, (say $r_2'(t)$), must be either on $\mathcal{B}_{UL}$ of $\mathcal{M}_{p-1}$ or on further left of it. In both cases, $r_2'(t) \notin \mathcal{IR}(\mathcal{R})$, which is a contradiction to Observation~\ref{obs:1}. Thus, $L$ must be strictly in between $L_1$ and $L_2$, and since there is no other vertical grid line between $L_1$ and $L_2$ (as the distance between $L_1$ and $L_2$ is 1), the line of symmetry does not pass through any grid nodes (Figure \ref{fig:smmetry reflectional horizon vertical DSG}). Similarly, we can show that if $C(t)$ admits a horizontal axis of symmetry, then it must not pass through any grid lines. Hence, we can conclude that if $C(t)$ admits a symmetry about a vertical or horizontal axis of symmetry and $G_M(\mathcal{R})$ is a disconnected step-graph, then $C(t)$ admits a partitive automorphism.
\end{proof}
\begin{figure} 
    \centering
\includegraphics[scale=0.09]{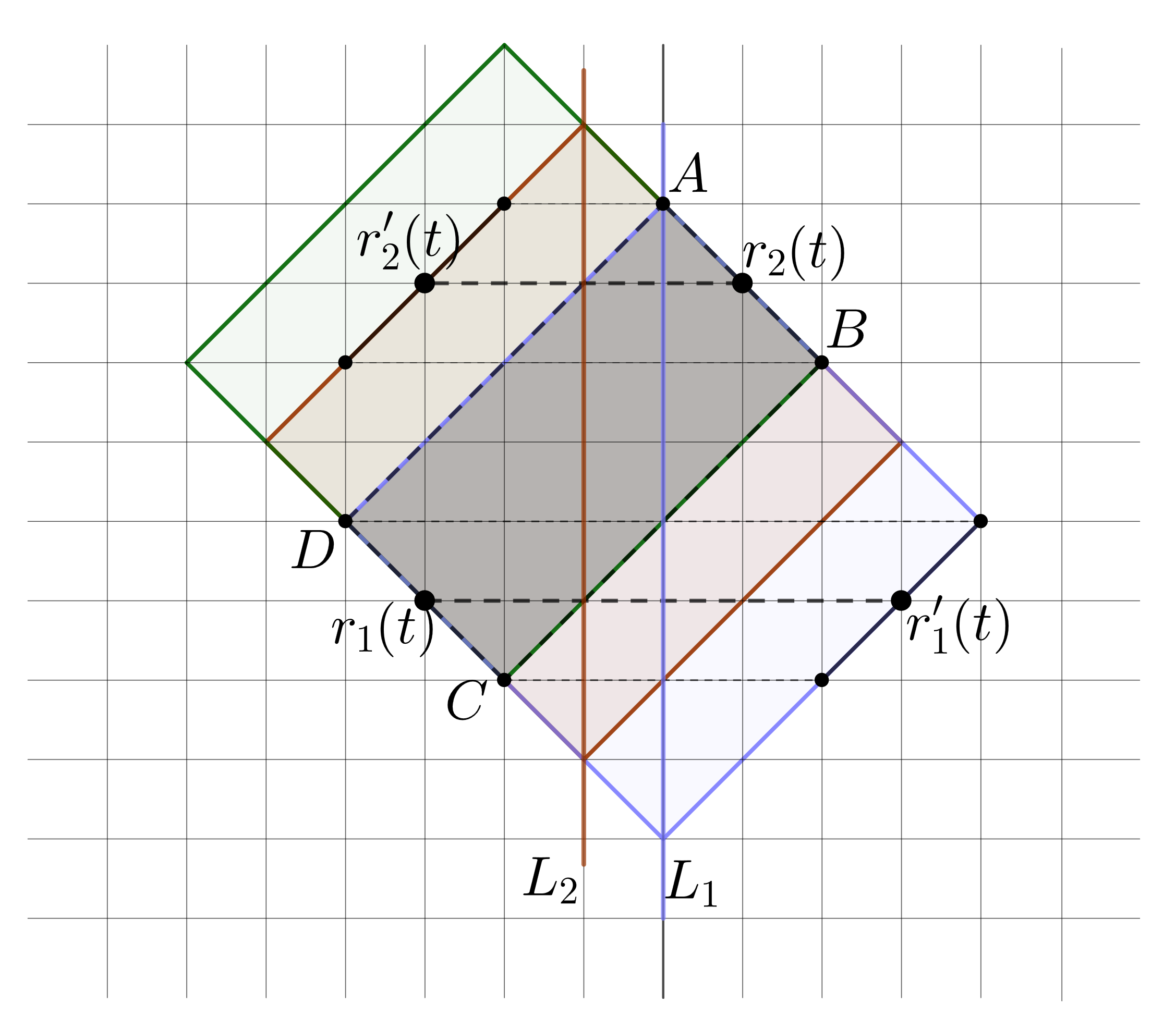}
    \caption{Reflection of a robot $r_1(t)$ on $CD$ about $L_1$ and reflection of a robot $r_2(t)$ on $AB$ about $L_2$ is outside of $\mathcal{I}_R(\mathcal{R})$.}
    \label{fig:smmetry reflectional horizon vertical DSG}
\end{figure}
\begin{lemma} \label{impossible4}
     If $C(t)$ is a configuration admitting rotational symmetry and such that $G_M(\mathcal{R})$ is a disconnected step graph, where the number of min-max nodes is even, then the configuration admits a partitive automorphism.
     \end{lemma}
     \begin{proof}
     Let $C(t)$ admit a rotational symmetry about a node $c$. That is, there exists an automorphism $\phi_{rot} \in Aut(C(t))$ such that $|<\phi_{rot}>| >1$. First, observe that if $C(t)$ admits a rotational symmetry, the intersection rectangle $\mathcal{IR}(\mathcal{R}) = ABCD$ remains the same even after the whole configuration is rotated by the angle of rotation about the center of rotation. So, the center of rotation of the configuration coincides with the center of the rectangle $ABCD$. Now, since $G_M(\mathcal{R})$ is a disconnected step-graph, $ABCD$ has sides $AD$ and $BC$ of length $k\sqrt{2}$ and sides $AB$ and $CD$ of length $(k-p+1)\sqrt{2}$, where $k\sqrt{2}$ is the size of the minimal enclosing diamonds in $C(t)$ and $p$ is the number of vertices in $G_M(\mathcal{R})$ (i.e., the number of min-max nodes in $C(t)$). Let us assume that $p =2p'$, for some $p' \in \mathbb{Z}$. Depending on whether $k$ is even or odd, the following two cases are to be considered.

         \textit{Case I:} Let $k$ be of the form $2k'$, for some $k' \in \mathbb{Z}$. Let $E$, $F$, $G$ and $H$ be the central nodes of the sides $AB$, $BC$, $CD$ and $DA$ respectively. Then $d_m(A,E) = d_m(B,E) = \frac{2k'\sqrt{2}}{2} = k'\sqrt{2}$. Since, both $A$ and $B$ are grid nodes, $E$ must also be a grid node. Similarly, $G$ is also a grid node. Now, $d_m(E,G) = d_m(D,A) = d_m(B,C) = (k-p+1)\sqrt{2} = (2k'-2p'+1)\sqrt{2}$, which is of the form $(2x+1)\sqrt{2}$ where $x = k'-p' \in \mathbb{Z}$. Now, the central node $m$ of this median $EG$ is the center of rotation. Thus $d_m(E,m) = d_m(m,G) = (x+\frac{1}{2})\sqrt{2}$. Since, $x+\frac{1}{2} \notin \mathbb{Z}$, $m$ can not be a grid node.

         \textit{Case II:} Let $k$ be of the form $2k'+1$, where $k' \in \mathbb{Z}$. Thus, $BC$ and $DA$ has lengths = $(2k'+1-2p'+1)\sqrt{2}$, which is of the form $2y\sqrt{2}$ where $y= k'-p'+1 \in \mathbb{Z}$. So, $d_m(D,H)= d_m(H,C)= y\sqrt{2}$. Now, since $A$ and $D$ both are grid nodes and $y \in \mathbb{Z}$, $H$ is also a grid node. Similarly, $F$ is also a grid node and the length of the median $HF$ is $k\sqrt{2} = (2k'+1)\sqrt{2}$. So, the central node of side $HF$ denoted as $m$ (which is also the center of rotation of $C(t)$) is $(k'+\frac{1}{2})\sqrt{2}$ distance away from both $H$ and $F$. Since $H$ and $F$ both are grid nodes and $k'+\frac{1}{2} \notin \mathbb{Z}$, $m$ is not a grid node.

         In both the cases, the center of rotation is not a grid node. Hence,  $C(t)$ admits a partitive automorphism. 
     \end{proof}

     Now, from Theorem~\ref{thm: impossiblePartitive}, we can state the following corollary.
\begin{corollary} \label{corollary3}
   If $C(0)$ is a configuration such that $G_M(\mathcal{R})$ is a disconnected step-graph, then min-max gathering is unsolvable if any one of the following holds true for $C(0)$.
   \begin{enumerate}
       \item $C(0)$ admits a reflective symmetry about an axis of symmetry, which is either vertical or horizontal. 
       \item $C(0)$ admits rotational symmetry about a node and $|V_M(\mathcal{R})|$ is even.
   \end{enumerate}
\end{corollary}
Next, consider the case when $G_M(\mathcal R)$ is a step-graph. We have the following lemma.
\begin{lemma}\label{lemma 22}
    Let $C(t)$ be a configuration such that the sub-graph induced by the set of min-max nodes of the configuration is a step-graph. If the step-graph has more than two nodes, and the configuration admits either a vertical or horizontal reflectional symmetry, then the configuration $C(t)$ must be partitive.
\end{lemma}
\begin{proof}
    Let $C(t)$ be a configuration such that the graph induced by the set of min-max nodes is a step-graph. It has been assumed that the step-graph contains more than two nodes. Note that the intersecting rectangle $\mathcal{IR}(\mathcal{R})$ of the set of robots $\mathcal R$ must contain at least one robot position at each of the boundaries and all the robot positions must be either inside or on the boundaries of $\mathcal {IR}(\mathcal R)$. Without loss of generality, assume that the configuration admits a vertical line of symmetry $l$ passing through the fixed line of the diamond. Consider a robot $r$ at one of the sides of the boundary. Note that, with respect to any vertical line of symmetry $l$, the equivalent robot of $r$, i.e., $\phi (r)$, must lie outside the rectangle $\mathcal {IR} (\mathcal R)$. This is a contradiction to the fact that all the robot positions must be either inside or on $\mathcal {IR}(\mathcal R)$. As a result, the configuration cannot admit a vertical line of symmetry passing through the fixed line. In case the configuration admits a horizontal line of symmetry, similar arguments can be provided. Hence, $C(t)$ cannot admit any vertical or horizontal lines of symmetry passing through the fixed line of the diamond, when there exist more than two nodes. Therefore, the line of symmetry must be such that the configuration is partitive.

\end{proof}
\begin{lemma}\label{lemma 23}
    Let $C(t)$ be a configuration such that the sub-graph induced by the set of min-max nodes of the configuration is a step- graph. If the step-graph contains at least two vertices and the configuration admits a diagonal line of symmetry passing through the fixed line of the intersecting rectangle, then $C(t)$ must be partitive. 
\end{lemma}
\begin{proof}
Let $C(t)$ be a configuration such that the sub-graph induced by the set of min-max nodes is a step-graph. Note that as the number of min-max nodes keeps on increasing in the configuration, there must exist precisely two sides of the boundary of the intersecting rectangle $\mathcal {IR}(\mathcal R)$ whose length of the side remains fixed. The two such boundary lines of the intersecting rectangle whose length remains fixed as the number of min-max nodes keeps on increasing are denoted as \textit{fixed lines} of the intersecting rectangle. We want to prove that if the configuration is such that the sub-graph induced by the set of min-max nodes is a step-graph and the configuration admits a diagonal line of symmetry passing through that fixed lines, then the configuration must always be partitive.

Note that if there exists such a diagonal line of symmetry that passes through the fixed lines, then that line must cross either the grid node of the fixed lines or of the center of two grid points of the fixed lines, i.e., passes through the edges of the grid (Figure \ref{fig:diagonalimpo}). Let $ABCD$ be the intersecting rectangle $\mathcal {IR}(\mathcal R)$, where $AB$ and $CD$ are the fixed lines of the rectangle. Note that the length of side $AB$ is $(k-\frac{1}{2})\sqrt{2}$, where $k\sqrt{2}$ is the side of the minimum enclosing diamond and obviously $k$ must be an integer. If there exists a diagonal line of symmetry passing through the fixed lines, then that line must pass through the mid point of $AB$. However, we can check that the length of $\frac{AB}{2}$ is $(\frac{k}{2}-\frac{1}{4})\sqrt{2}$. If the mid point of $AB$ is grid node, then $(\frac{k}{2}-\frac{1}{4})$ must be an integer. So, if $(\frac{k}{2}-\frac{1}{4})= z$, then $k=2z+\frac{1}{2}$ must be an integer. However, $k$ is not an integer and hence it leads to a contradiction. So, there does not exist a diagonal line of symmetry passing through the fixed lines in $C(t)$. Hence, it can be concluded that the configuration is partitive.
\end{proof}
\begin{figure} 
    \centering
\includegraphics[scale=0.35]{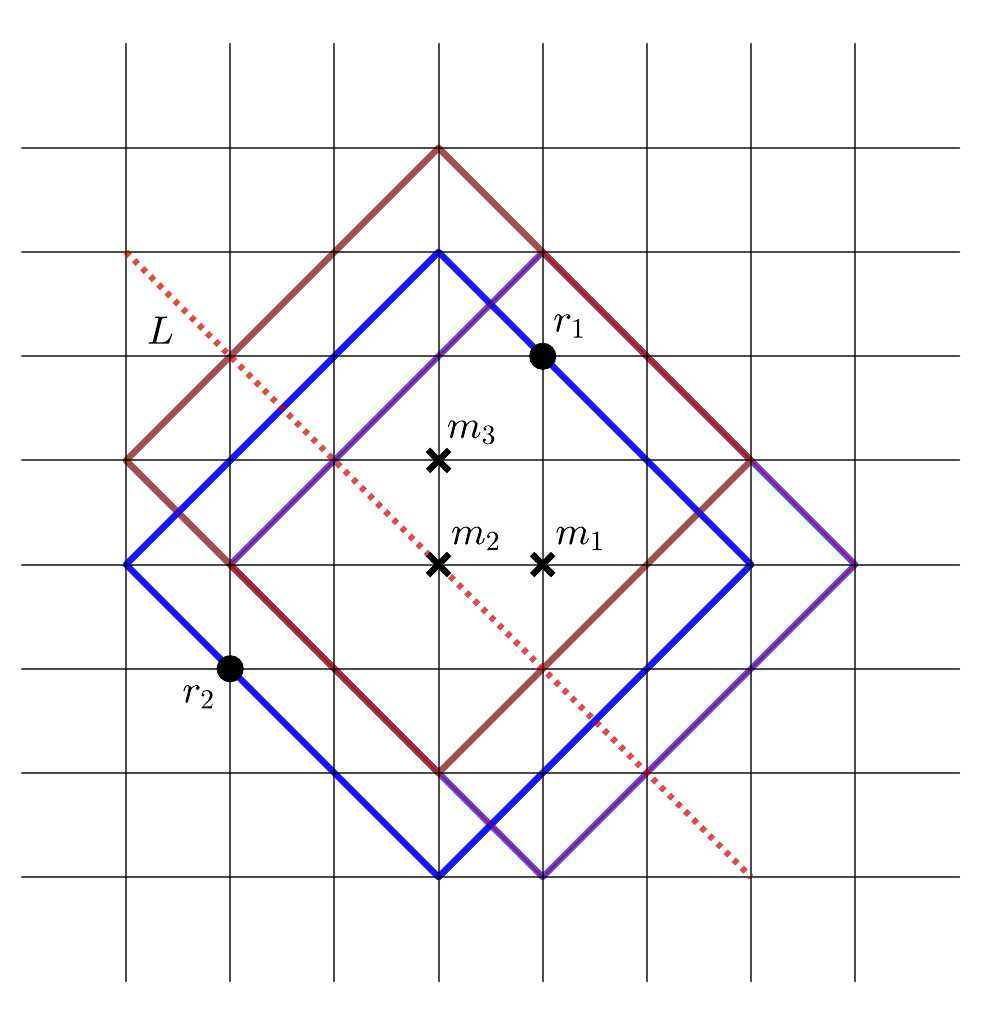}
    \caption{Reflection of a robot $r_1$ about the line of symmetry $L$ is $r_2$ lying outside $\mathcal {IR}(\mathcal R)$, which is a contradiction.}
    \label{fig:diagonalimpo}
\end{figure}
\begin{lemma}\label{lemma 24}
    Let $C(t)$ be a configuration such that the sub-graph induced by the set of min-max nodes of the configuration is a step-graph and the configuration admits rotational symmetry. If the step-graph contains at least two vertices, then the configuration $C(t)$ must be partitive. 
\end{lemma}
\begin{proof}
   Let $C(t)$ be a configuration such that the sub-graph induced by the set of min-max nodes is a step-graph. It has also been assumed that the initial configuration admits a rotational symmetry. If the configuration is not partitive, the configuration must be such that the centre of rotation must be a grid node. It should also be noted that if the centre of rotation is a grid node, then it must be the intersection node of the lines passing through the sides of the intersecting rectangle. However, in the proof of the previous lemma \ref{lemma 23}, it has been proved that there do not exist such lines that pass through the grid node of the fixed lines. So, it can be concluded that the intersection node is also not a grid node. Hence, if the configuration is such that the sub-graph induced by the set of min-max nodes is a step-graph and it admits rotational symmetry, then the configuration must be partitive. 
\end{proof}
\begin{corollary} \label{corollary4}
    If the initial configuration $C(0)$ is such that $V_M(\mathcal{R})$ is a step-graph, then min-max gathering is unsolvable if $C(0)$ is any one of the following:
  \begin{enumerate}
      \item $C(0)$ admits vertical or horizontal reflectional symmetry when the step-graph has more than two vertices.
      \item $C(0)$ admits diagonal reflectional symmetry passing through the fixed line of the intersecting rectangle.
      \item $C(0)$ admits rotational symmetry about a node.
     
  \end{enumerate}
\end{corollary}

\section{Algorithm} \label{sec:algorithm}
In this section, a deterministic distributed algorithm \textit{Gathering()} $\mathcal A$ has been proposed to solve the min-max gathering problem in infinite grids. Let $\mathcal U$ be the set of all ungatherable configurations listed in the Corollaries \ref{corollary2}, \ref{corollary3} and \ref{corollary4}. The algorithm $\mathcal A$ has been proposed to gather the robots at one of the min-max nodes, for all the initial configurations in $\mathcal I \setminus \mathcal U$ and for all configurations with the number of robots at least nine. According to our analysis, the subgraph induced by the set of min-max nodes $G_M(\mathcal R)$ can take one of three forms: a cycle with four nodes, a disconnected step-graph, or a connected step-graph, as defined in Definitions 12 and 13. Therefore, the proposed algorithm is in accordance with the structure of the initial configuration. The main crux of the algorithm is to select a min-max node (say $m$) among all the possible min-max nodes and allow the robots to move towards $m$, such that the robots uniquely identify $m$ as the target min-max node. However, we have seen examples where it has been observed that a min-max node may not remain invariant under the movement of the robot towards itself. The main difficulty lies in selecting the target node, where the gathering must be ensured. The proposed algorithm is divided into the following phases:
\begin{enumerate}
    \item \textit{Target Min-Max Node Selection phase}: In this phase, the target min-max node (say $m$) is selected to ensure gathering on that node.
    \item \textit{Creating multiplicity phase}: In this phase, the robots create a multiplicity at the target $m$ in order to ensure gathering on that node.  
    \item \textit{Finalization phase}: In this phase, the gathering is finalised at $m$. Since the robots have global-weak multiplicity detection capability, the guards can identify the unique multiplicity node $m$ and finalize the gathering at $m$.
\end{enumerate}
  We will provide a detailed explanation of the proposed algorithm in the next subsections. Before doing so, we have categorized the possible initial configurations. All the initial configurations can be partitioned into the following distinct, non-overlapping classes:
 
  \begin{figure}
  
    \centering
\includegraphics[scale=0.39]{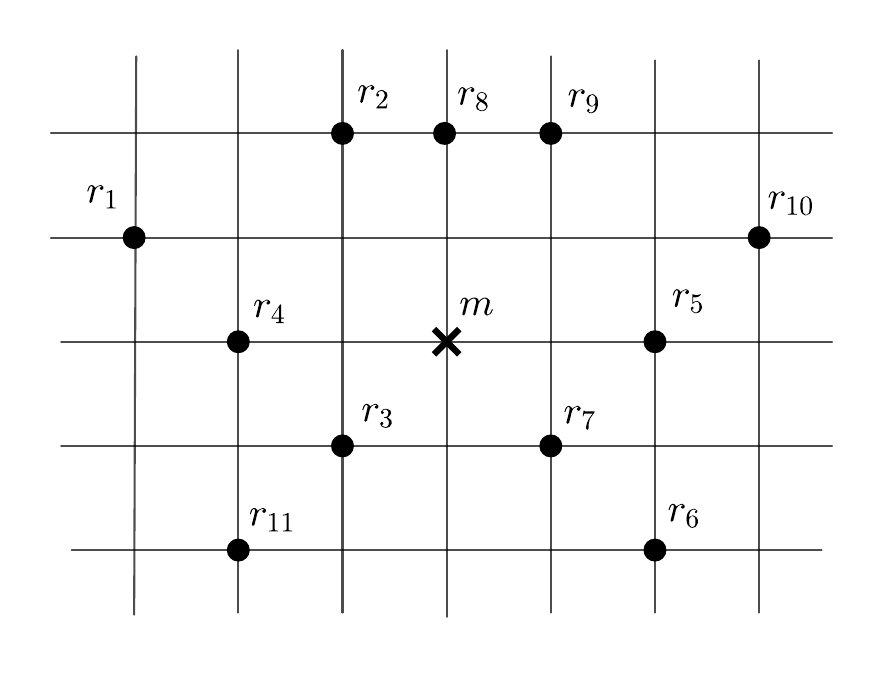}
    \caption{Example of an $\mathcal I_1$ configuration. Configuration admits only one min-max node $m$.}
    \label{I1}
\end{figure}
\begin{enumerate}
    \item $\mathcal I_1$: All configurations for which $|V_M(\mathcal R)|=1$ (Figure \ref{I1}).
    \item $\mathcal I_2$: All configurations for which $|V_M(\mathcal R)| \geq 2$ and the configuration is asymmetric (Figure \ref{I2}).
    \begin{figure}
    \centering
\includegraphics[scale=0.32]{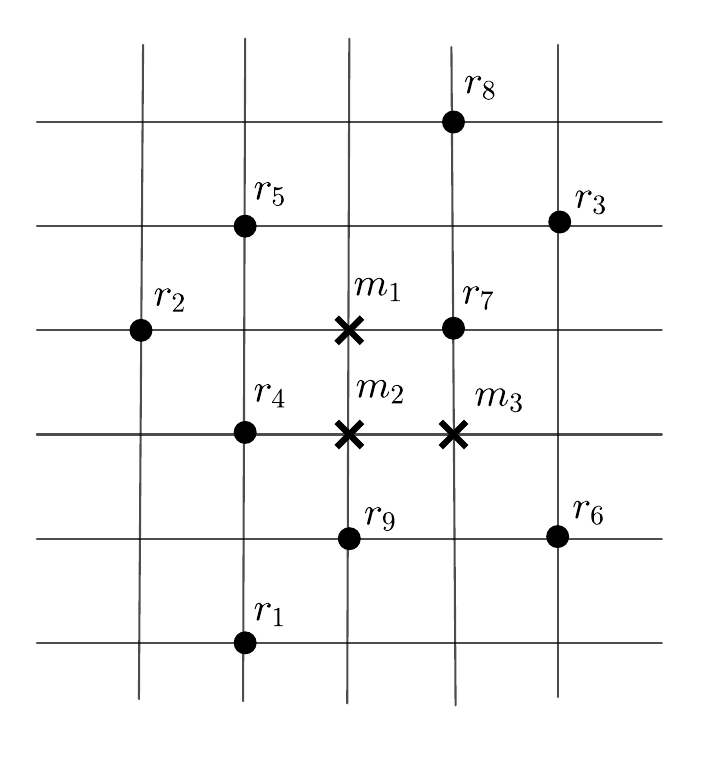}
    \caption{Example of an $\mathcal I_2$ configuration. The configuration is asymmetric and admits three min-max nodes $m_1$, $m_2$ and $m_3$.}
    \label{I2}
\end{figure}
    \item $\mathcal I_3$: All configurations for which $|V_M(\mathcal R)| \geq 2$ and the configuration is symmetric with respect to a single reflectional line of symmetry (Figure \ref{I3}).
    \begin{figure}
    \centering
\includegraphics[scale=0.32]{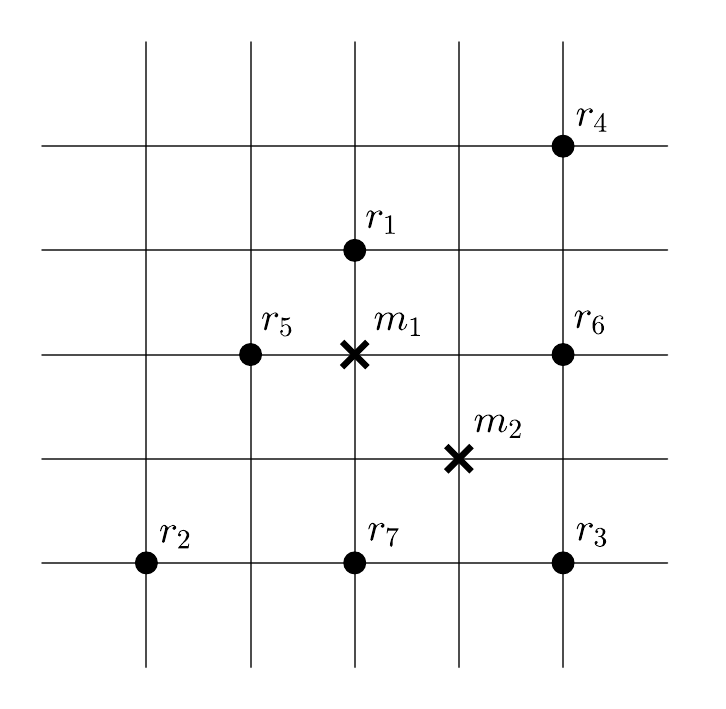}
    \caption{Example of an $\mathcal I_3$ configuration. The configuration is symmetric with respect to a diagonal line of symmetry and admits two min-max nodes $m_1$ and $m_2$.}
    \label{I3}
\end{figure}
    \item $\mathcal I_4$: All configurations for which $|V_M(\mathcal R)| \geq 2$ and the configuration is symmetric with respect to a rotational symmetry (Figure \ref{I4}).
\end{enumerate}
We consider the following observations.
\begin{observation} \label{observation5}
    If the configuration is asymmetric, the robots can reach an agreement on a common coordinate system \cite{DBLP:journals/tcs/BoseAKS20}.
\end{observation}
\noindent Similarly, we have the following observation.

\begin{observation} \label{observation6}
    If the configuration is symmetric with respect to a single line of symmetry $l$, the nodes on $l$ can be ordered uniquely.
\end{observation}
Next, we consider the following definitions relevant to the selection of the target min-max node in the initial configuration.
\begin{definition}
Weber min-max node: A min-max node $m$ is defined as a Weber min-max node if it minimizes the value $c_t(m)$, where $c_t(m)$ denotes the centrality of the node $m$ at time $t$.
\end{definition}
\begin{figure}
    \centering
\includegraphics[scale=0.39]{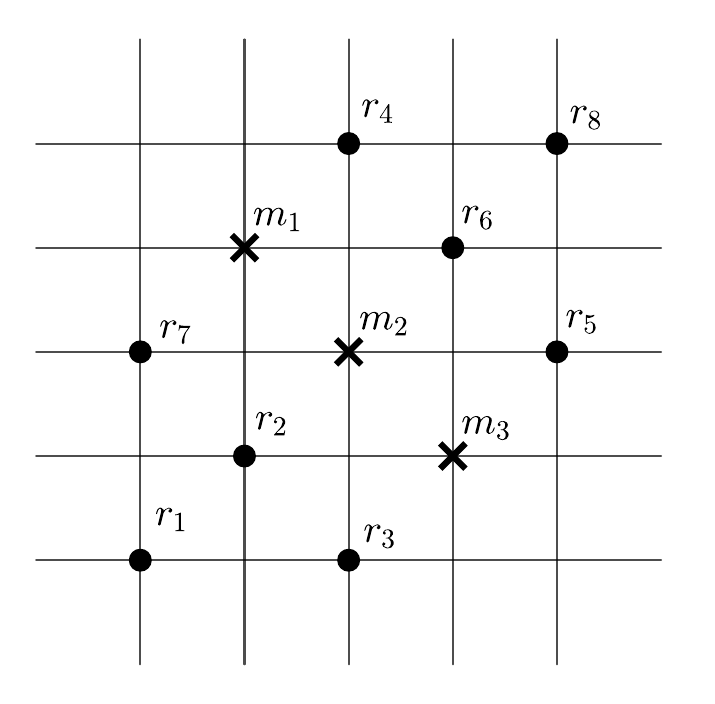}
    \caption{Example of an $\mathcal I_4$ configuration. The configuration is symmetric with respect to rotational symmetry and admits three min-max nodes $m_1$, $m_2$ and $m_3$.}
    \label{I4}
\end{figure}
\begin{definition}
    Co-boundary of a diamond \Diam: For a diamond \Diam\ of size $k\sqrt{2}$ ($k \in \mathbb{N}$), its co-boundary is the set of all grid nodes $$CB(\mathcal M)= \{v \in V: d_m(v,\mathcal{C}(\mathcal{M})) = k-1\}.$$
\end{definition}
That is, the co-boundary of a diamond \Diam\ is defined as the set of all nodes that are at a distance of $k$-1 from the center of \Diam. 

\begin{definition}
Let $CB(\mathcal{M})$ be a co-boundary of a diamond $\mathcal{M}$. A side $S$ of a $CB(\mathcal{M})$ corresponding to a boundary $\mathcal B$ of $\mathcal{M}$ is a set $S_{CB}(\mathcal B)$ =\{$v\in CB(\mathcal{M})$: $\exists v'\in \mathcal B $ for which $d_m(v,v')=1$\}
\end{definition}
We next consider the definition of a movable robot.
\begin{definition}
   Suppose that the condition $(C1)$ from Lemma~\ref{lemma: condition for unique minmax point} holds. A robot is said to be movable if it satisfies the following conditions.
  \begin{enumerate}
    \item Any interior robot
    \item A robot that lies within a boundary of \Diam\ and does not have the maximum view among all the robots on the same boundary.
\end{enumerate}
Suppose that the condition $(C2)$ from Lemma~\ref{lemma: condition for unique minmax point} holds. A robot is said to be movable if it satisfies the following conditions.

\begin{enumerate}
    \item An interior robot not lying on the co-boundary.
    \item An interior robot lying on some side $S$ of the co-boundary corresponding to a boundary of the diamond for which there exist robots on that specific boundary.
    \item A robot that lies on some side $S$ of the co-boundary corresponding to a boundary $B$ of the diamond such that $B$ has no robot positions and the robot under consideration does not have the largest view among all such robots on $S$.
    \item A robot on the side of the boundaries, provided that it is not the robot with the maximum configuration view among all such robots that lie within that particular side of the boundary.
\end{enumerate}
Suppose that the condition $(C3)$ from Lemma~\ref{lemma: condition for unique minmax point} holds. 
 A robot is said to be movable if it satisfies the following conditions.

\begin{enumerate}
    \item An interior robot not lying on the co-boundary.
    \item A robot positioned on a side $S$ of the co-boundary for which robots are located at the boundaries at a distance one from the side $S$.
    \item A robot that lies on the sides of co-boundaries, such that the boundaries at a distance one from the co-boundaries contain no robot positions and the robots under consideration do not have the maximum configuration view among all such robots on the sides.
    \item A robot on the boundaries, provided that it is not the robot with the maximum configuration view among all robots that lies within that particular boundary.
\end{enumerate}
\end{definition}
We next proceed to the description of the algorithm. The algorithm proceeds according to the partition of the initial configurations.
\subsection{$\mathcal I_1$}
Consider the case when $|V_M(\mathcal R)|=1$. This implies there exists a unique min-max node in the initial configuration. First, assume the case when the initial configuration is asymmetric. According to Lemma \ref{lemma: condition for unique minmax point}, if there exists a unique min-max node, there are three possible cases that are to be considered. The algorithm is described according to the conditions $(C1)$, $(C2)$ and $(C3)$ that an initial configuration can satisfy and is listed in the Lemma \ref{lemma: condition for unique minmax point}. 
\subsubsection{Target Min-Max Node Selection Phase}
The unique min-max node (say $m$) is selected as the target min-max node.
\subsubsection{Creating Multiplicity Phase}
This phase of the algorithm is described according to the following cases.

\paragraph {Case 1} \label{case1} Assume that the condition $(C1)$ from Lemma~\ref{lemma: condition for unique minmax point} holds. This implies there exist boundary robots at each boundary of $\mathcal {MED} (\mathcal R)$. The following subcases are to be considered.

\paragraph {Subcase 1} \label{subcase1} Number of interior robots is at least two. A movable robot that is closest to the unique min-max node $m$ is allowed to move towards $m$. If there are multiple such robots, consider the robot that is closest and has the minimum view. While the robot reaches $m$, consider the interior robot, which is closest to $m$ and is not on $m$. If there are multiple such robots, consider the closest robot $r$ with the minimum view. $r$ moves towards $m$ and creates a multiplicity at $m$. 

\paragraph{Subcase 2} \label{subacse2} There exists a unique interior robot. The unique interior robot $r$ starts moving towards $m$. While the robot is at a distance one from $m$, it stops moving. A movable robot on the boundary, which is closest and has the minimum configuration view in case of ties, starts moving towards $m$. Note that such a robot always exists as $n \geq 9$. While the robot moves towards an adjacent node, it becomes an interior robot. The rest of the procedure follows similarly as in Subcase 1. 

\paragraph{Subcase 3} There exists no interior robot. Consider a movable robot from the boundary. If there are multiple such movable robots, consider the movable robot with the minimum configuration view. Note that such a robot always exists as $n \geq 9$. As the robot moves towards an adjacent node, it becomes an interior robot, and the case proceeds similarly to that in Subcase 1.

\paragraph{Case 2} \label{case2} 
Assume that the condition $(C2)$ from Lemma~\ref{lemma: condition for unique minmax point} holds. This implies that there exists at least one boundary robot at the three sides of the boundaries of $\mathcal {MED} (\mathcal R)$. Furthermore, at least one robot exists on the co-boundary, corresponding to which no boundary robot exists. The following subcases are to be considered. 

\paragraph{Subcase 1} \label{subcase11} 
Number of interior robots is at least two. First, consider the case when there exist at least two interior robots that lie inside the co-boundaries. Consider such an interior robot that is closest to $m$. If there are multiple such robots, consider the robot with the minimum configuration view. The robot is allowed to move towards $m$. Another such interior robot (say $r$), which is closest to $m$, starts moving towards $m$ and creates a multiplicity at $m$. No other robot is allowed to move unless $r$ reaches $m$. Consider the scenario in which there is no more than one interior robot that lies within the co-boundary. If there exists exactly one interior robot lying inside the co-boundary, the robot starts moving towards $m$. While the robot is at a distance one from $m$, the robot stops moving. Next, the robot on any co-boundary first checks whether it is movable or not, i.e., whether it is the only robot on that particular co-boundary and at a distance of one unit from the boundary, not containing any robot position. If it is movable, then it starts moving towards $m$ and the procedure proceeds similarly to the previous case. Thus, there exists an instant of time $t>0$ at which a multiplicity is created at $m$. Otherwise, if the robot on the co-boundary is not movable, then a movable robot is considered on the boundary and it moves towards $m$. Thus, the robots create a multiplicity at $m$. If more than one such robot exists, the ties are broken by considering the robot with the minimum configuration view. Note that such a robot always exists as we have assumed that in the initial configuration $n \geq 9$. 

\paragraph{Subcase 2} Number of interior robots is exactly one. In this case, the interior robot must lie on the co-boundary. Note that the interior robot is not movable. Thus, the robot on the co-boundary will not move. As it has been assumed that the number of robots $n \geq 9$, there must exist at least two boundaries of $\mathcal {MED} (\mathcal R)$ that contain robots that are movable. The boundary robot, which is movable, and having the minimum configuration view in case of ties, is selected and starts moving towards $m$. When such a robot is at a distance one from $m$, the other robot from the boundary which is movable, moves towards $m$, and when it becomes an interior robot, a multiplicity creation procedure initiates at $m$. The rest of the procedure follows from Subcase 2 of Case 1.

\paragraph{Subcase 3} There exists no interior robot. Note that this case is not possible, as at least one robot must exist at one of the co-boundaries.

\paragraph{Case 3} \label{case3} 
Assume that the condition $(C3)$ from Lemma~\ref{lemma: condition for unique minmax point} holds. As a consequence, only two opposite boundaries contain robot positions, and at least one robot position lies within a distance of one from each boundary not containing a robot position. The following subcases are to be considered.

\paragraph{Subcase 1} There exist at least two interior robots. First, consider the case when there exist at least two robots that are not lying on the co-boundaries. In this case, first the interior robot closest to $m$ and having the minimum view move towards $m$. Next, the other interior robot not lying on the co-boundary and closest to $m$, moves towards $m$, creating a multiplicity at $m$. If there are multiple such robots, the ties are broken by considering the robot with the minimum configuration view. Next, consider the case when there exists exactly one interior robot not lying on any co-boundary. The unique robot moves towards $m$ and when it is at a distance one from $m$, it stops moving. The robots on the co-boundaries first check whether it is movable or not. If the robots on the co-boundaries are movable, then the closest robot and minimum view in case of ties move towards an adjacent node towards $m$. Thus, there exist at least two interior robots at some time instant $t>0$ and the procedure proceeds similarly to before. If the robots on the co-boundaries are not movable, then the robots on the boundaries, which are movable, move towards $m$ in a sequential manner, according to their distances from $m$. Thus, a multiplicity is created at $m$ at some time instant. The rest of the procedure follows similarly as in Subcase 1 of Case 1.

\paragraph{Subcase 2} There exists at most one interior robot. Note that this case is not possible, as robots must exist on two co-boundaries.
\subsubsection{Finalisation}
The other robots that are not on $m$ can identify the multiplicity at $m$ due to the presence of weak multiplicity-detection capability of the robots. As a result, the other robot moves towards $m$ without creating any other multiplicity. According to Observation \ref{observation5}, since the configuration is asymmetric, the robots are orderable. The robot that is not on $m$ and closest to $m$ is selected as the candidate robot to move towards $m$. If there are multiple such robots, the ties can be broken by considering the robot with the minimum view. Consequently, at any instant of time, a unique robot can be chosen to move towards $m$. Thus, no other multiplicity is created and gathering is finalised at $m$.

\subsubsection{Symmetric cases} 

\noindent \textbf{Configuration is symmetric with respect to a single line of symmetry: } Next, consider the case when the configuration is symmetric with respect to a single line of symmetry (say $l$). According to the Lemma \ref{lemma 16}, there must exist at least one robot position or min-max node on $l$. First, assume that there exists at least one robot position on $l$, and the robot is movable. Here, a robot is said to be movable if there exists an empty node away from $l$, which is adjacent to the robot. If there exist multiple such movable robots on $l$, the movable robot with the minimum configuration view on $l$ moves towards an adjacent node away from $l$. This results in transforming the configuration into an asymmetric configuration, and the procedure follows as before in the case of asymmetric configurations. Thus, consider the case when no robot exists on $l$. This implies that the unique min-max node $m$ is on $l$. $m$ is selected as the target min-max node. Next, the following cases are to be considered.

\paragraph{Case 1} The line of symmetry is either horizontal or vertical. Note that in this case, there must exist at least one robot on each boundary. First, assume the case, when there exist at least two interior robots. In this case, the movable robots that are closest to $l$, are allowed to move towards $m$. If there are multiple such robots, consider the robots that are closest and have the minimum views. Note that there are exactly two such robots. The two robots may either move concurrently or there may be a possible pending move leading to asymmetry of the configuration. If the configuration becomes asymmetric, then the procedure follows similarly as before in asymmetric cases. When the robot reaches $m$, it creates a multiplicity at $m$. Note that if there exists exactly one interior robot, then this robot must lie on $l$ and the robot is movable. The robot moves along the line towards $m$. The two robots that are closest and with minimum views move towards $m$ and the procedure proceeds similarly as before. Next, consider the case when there are no interior robots. Consider a movable robot on the boundary of $\mathcal {MED (R)}$ and its equivalent robot with respect to $l$. If there are multiple such movable robots, then consider the movable robots with the minimum configuration view. As the robots move towards an adjacent node, they become interior robots, and the case proceeds similarly to before when there are at least two interior robots.


\paragraph{Case 2} The line of symmetry is a diagonal line of symmetry. First, assume that there exists at least two interior robots. Consider the case, when there exist at least two interior robots that are not lying on the co-boundaries. The procedure follows similarly as in the case when the line of symmetry is either a horizontal or vertical line of symmetry. If there exist robots only on the co-boundaries, the robots on the co-boundaries first check whether they are movable or not, i.e., whether they are the only robots lying on the side of the co-boundaries to which they belong. If it is the unique robot, then it is not movable. Otherwise, the robot on the co-boundary with the minimum configuration view and its symmetric image moves towards $m$ and creates a multiplicity at $m$. Note that the configuration may become asymmetric due to the asynchronous behavior of the scheduler.

In the finalisation phase, the other robots that are not on $m$ can identify the multiplicity at $m$ due to the weak multiplicity-detection capability of the robots. The robots that are not on $m$ and closest to $l$ start moving towards $m$ (say $r$ and $\phi(r)$). If there are multiple such robots, the ties are broken by considering the robot with the minimum views. While moving towards $m$, $r$ and $\phi(r)$ are first aligned along the same grid-line containing $m$ and then start moving towards $m$. As a result, the other robots can move towards $m$ without creating any other multiplicity other than $m$. Thus, the gathering is finalised at $m$.

\noindent \textbf{Configuration is symmetric with respect to rotational symmetry:} The target min-max node is selected as the node $c$, which is the center of rotational symmetry. The point $c$ is also the unique Weber node of the configuration. In this scenario, all boundary positions must include robot positions. As a result, all the robots agree to gather at the unique Weber node, which also serves as the configuration's unique min-max node. All the movable robots that are closest to $m$ start moving towards $m$. If there are multiple such closest robots, the ties are broken by considering the closest robots with minimum views. Note that at any instant of time, only four robots are allowed to move towards $m$. As $m$ is a Weber node, $m$ remains invariant during the movement of the robots (as a Weber node remains invariant while a robot moves towards it). Thus, the gathering is finalised at $m$.

\subsection{$\mathcal I_2$}
Consider the case, when $C(0) \in \mathcal I_2$. This implies that $|V_M(\mathcal R)| \geq 2$ and the configuration is asymmetric. The phases of the algorithm are described in the next subsections.  
\subsubsection{Target Min-Max Node Selection Phase}
First, assume the case when the subgraph $G_M(\mathcal R)$ induced by the set of min-max nodes is a disconnected step-graph. 
If $|V_M(\mathcal{R})|$ is odd, the target min-max node is selected as the unique central node of the subgraph (say $m$) induced by the set of min-max nodes. Next, if $|V_M(\mathcal{R})|$ is even, there always exist two nodes (say $m_1$ and $m_2$), which are the central nodes of the subgraph induced by the set of min-max nodes. The target min-max node is selected as the min-max node between $m_1$ and $m_2$. First, among the min-max nodes, $m_1$ and $m_2$, if there exists a unique Weber min-max node, the robots select that node as the target min-max node. Otherwise, if there does not exist a unique Weber min-max node, then each robot calculates the distance of $m_1$ and $m_2$ to its closest robot. Assume that the distance of $m_1$ (resp. $m_2$) to its closest robot is $d_1$ (resp. $d_2$). First, assume the case when $d_1 \ne d_2$. Without loss of generality, let us assume that $d_1 >d_2$. The min-max node for which the distance $d_i, i=\lbrace 1, 2 \rbrace$ is smaller is chosen as the target min-max node ($m_2$ in our case). Next, consider the case when $d_1 = d_2$. As the configuration is asymmetric, the min-max nodes are orderable according to Observation \ref{observation5}. The min-max node with the minimum view among $m_1$ and $m_2$ is selected as the target min-max node. 

Next, consider the case when the subgraph induced by the set of min-max nodes is a step-graph. Note that the subgraph is a path graph. The target min-max node is selected similarly as in the case when $G_M(\mathcal R)$ is a disconnected step-graph. 

Finally, assume the case when the subgraph induced by the set of min-max nodes is a four-cycle. Let the four min-max nodes $m_1$, $m_2$, $m_3$ and $m_4$ constitute the four cycles. Among the min-max nodes, first, consider the case when there exists a unique Weber min-max node (say $m_1$). In this case, $m_1$ is selected as the target min-max node. Otherwise, there is more than one Weber min-max node. In that case, the robots first check the distance of the closest robots from each of these Weber min-max nodes. If there exists a unique min-max node among such Weber min-max nodes for which the distance to its closest robot is minimum, then that min-max node is selected as the target min-max node. In case there are multiple such Weber min-max nodes, the target is selected as the min-max node, which has the minimum view among such min-max nodes. Note that since the configuration is asymmetric, there always exists such a min-max node according to Observation \ref{observation5}. 
\subsubsection{Creating Multiplicity Phase} 
First, assume the case when $G_M(\mathcal R)$ is a disconnected step-graph. The following cases are to be considered.

\paragraph{Case 1} \label{case111}
There are at least two interior robots. Here, an interior robot is a robot that lies inside the intersecting rectangle $\mathcal {IR}(\mathcal R)$. In this case, the interior robot that is closest to the target is allowed to move towards the target. If there is more than one such closest robot, the robot with the minimum view among such robots moves towards the target. If the number of central nodes is one, the target remains invariant during the movement of the robots. Otherwise, when the number of central nodes is two, the target might change. However, the algorithm ensures that during the execution of the algorithm, the target remains among $m_1$ and $m_2$. When the robot reaches one of the chosen targets, the target remains invariant during the execution of the algorithm. The robot, which is closest to the target and not on the target, moves towards the target. If there are multiple such closest robots, the ties are broken by considering the closest robot with the minimum view. Thus, a multiplicity is created at the target at some time instant $t>0$.

\paragraph{Case 2} There is exactly one interior robot. First, the interior robot moves towards the target, and when it is one move away from the target, it stops moving. Next, the robot on the boundary, which is movable and has the minimum configuration view among all such robots moves towards $m$. While it moves towards an adjacent node, it becomes an interior robot and the case proceeds similarly to that in Case 1.
 
\paragraph{Case 3} There are no interior robots. The robot on the boundary of the intersecting rectangle $\mathcal {IR}(\mathcal R)$, which is movable and has the minimum view in case of ties, moves towards the target. When it moves towards an adjacent node, it becomes an interior robot. The rest of the procedure follows similarly as in Case 2. 

Next, assume the case when the subgraph $G_M(\mathcal R)$ induced by the set of min-max nodes is a four-cycle graph. The following cases are to be considered. 

\paragraph{Case 1} There are more than two interior robots. The interior robot that is closest to the target min-max node (say $m$) starts moving towards $m$ until it reaches the target. If there are multiple such closest robots, the closest robot with the minimum view moves towards $m$. Note that while the closest robot starts moving towards the target, the target might change. However, our proposed algorithm guarantees that the target remains among one of the four nodes $m_1$, $m_2$, $m_3$ and $m_4$. The target is selected as the Weber min-max node with the minimum distance to its closest robot and with the minimum view in case of ties. When the robot reaches one of such $m_i's$, the target remains invariant during the execution of the algorithm. The next closest robot moves towards the target and creates a multiplicity at the target. 
    
\paragraph{Case 2} There is exactly one interior robot. As in the initial configuration, we have assumed that $n \geq 9$, a boundary of the intersecting rectangle $\mathcal {IR}(\mathcal R)$, must exist that has more than one robot on it. Since the configuration is asymmetric, a unique robot can always be selected from any boundary containing multiple robot positions, which is movable. The selected robot starts moving towards the interior of the intersecting rectangle and becomes an interior robot. Note that, while the robot starts moving from the boundary towards the target, the asymmetry of the configuration remains invariant. When the robot reaches the interior, the rest of the procedure follows from Case 1. 
    
  \paragraph{Case 3} There exist no interior robots. Since $n \geq 9$, there must exist at least two boundaries of the intersecting rectangle $\mathcal {IR} (\mathcal R)$ that has more than one robot on it. As the configuration is asymmetric, a robot can always be selected from a boundary containing multiple robot positions that are movable. The selected robot moves towards the target and enters inside the intersecting rectangle, thus becoming an interior robot. The rest of the procedure follows from the previous case.

  \noindent Finally, assume the case when $G_M(\mathcal R)$ is a step-graph. In this case, the procedure follows similarly as in the case when $G_M(\mathcal R)$ is a disconnected step-graph.

\subsubsection{Finalization}
Suppose there exists a multiplicity $m$ at the target min-max node. All the other robots that are not on the multiplicity $m$ will try to move towards $m$ in order to finalise the gathering process. As the configuration is asymmetric, the robots are orderable according to Observation \ref{observation5}. The interior robots that are not on $m$ will first move towards $m$ sequentially and according to their distances from $m$. Ties are broken by considering the robot with the minimum view. As a result, at each instant of time, only one robot is allowed to move towards the target. When all the interior robots reach $m$, the boundary robots that are movable start moving towards $m$. The boundary robots will move first to the interior in a sequential manner, according to the distance from the target. In this way, one by one, all boundary robots will move to the multiplicity node, and gathering will be finalised.


\subsection{$\mathcal I_3$}
Consider the case, when $C(0) \in \mathcal I_3$, this implies that $|V_M(\mathcal R)| \geq 2$ and the configuration is symmetric with respect to single line of symmetry $l$ (say).

First, assume that the configuration admits either a horizontal or vertical line of symmetry and $V_M(\mathcal R)$ is either a disconnected step-graph or a four-cycle. Note that according to Lemma~\ref{lemma 18} and Lemma~\ref{lemma 20}, the configuration must be partitive. Thus, whenever the configuration is such that $V_M(\mathcal R)$ is either a disconnected step-graph or a four-cycle, the configuration can only admit a diagonal line of symmetry. Next, assume that the configuration admits either a horizontal or vertical line of symmetry and $V_M(\mathcal R)$ is a step-graph. According to Lemma~\ref{lemma 22}, the configuration must be partitive. Lemma \ref{lemma 23} also states that if the step-graph contains at least two vertices and the configuration admits a diagonal line of symmetry passing through the fixed line of the
configuration, then also $C(t)$ must be partitive. This implies that if the configuration admits a diagonal line of symmetry and $V_M(\mathcal R)$ is a step-graph, then the diagonal line must not pass through the fixed line of the configuration. We next proceed to the algorithm description for $\mathcal I_3$. 
\subsubsection{Target Min-Max Node Selection Phase}
First, assume that the configuration is such that $V_M(\mathcal R)$ is a four-cycle, and the configuration admits a diagonal line of symmetry. Note that there exist exactly two min-max nodes on the line of symmetry $l$ (say $m_1$ and $m_2$). If there exists a unique Weber min-max node between $m_1$ and $m_2$, then that min-max node is selected as the target min-max node. Otherwise, both $m_1$ and $m_2$ are Weber min-max nodes. Each robot calculates the distance of $m_1$ and $m_2$ to its closest robot. If there exists a unique min-max node whose distance from it to its closest robot is minimum, then that min-max node is selected as the target min-max node. Otherwise, the target min-max node is selected as the min-max node on the line of symmetry $l$ with the minimum view. Note that, according to Observation \ref{observation6}, the nodes on $l$ can be uniquely orderable. 

Next, assume the case, when $V_M(\mathcal R)$ is a disconnected step-graph and the configuration admits a diagonal line of symmetry. The target min-max node is selected as the min-max node among $m_1$ and $m_2$, which are the central nodes of the subgraph induced by the set of min-max nodes. It may also be possible that there is only one central node in case the path induced by the set of min-max nodes contains an odd number of nodes. The target is selected among $m_1$ and $m_2$. The procedure follows similarly, as in the case when $V_M(\mathcal R)$ is a four-cycle.

Finally, assume the case, when $V_M(\mathcal R)$ is a step-graph. In this case, the configuration can admit only a diagonal line of symmetry not passing through the fixed line of the intersecting rectangle $\mathcal {IR} (\mathcal R)$. The line of symmetry can admit at most two min-max nodes. If there exists a unique min-max node on $l$, then that min-max node is selected as the target min-max node. Otherwise, $l$ contains two min-max nodes (say $m_1$ and $m_2$). The procedure follows similarly, as in the case when $V_M(\mathcal R)$ is a four-cycle.

\subsubsection{Creating Multiplicity Phase}
Let $m$ be the target min-max node selected in the Target Min-Max Node Selection phase. First, assume the case when $G_M(\mathcal R)$ is a disconnected step-graph. The following cases are to be considered.

\paragraph{Case 1} There are at least two robots inside the intersecting rectangle $\mathcal {IR} (\mathcal R)$. The interior robot and its symmetric image, which are closest to $v$, move towards $v$. In the case of a tie, the closest robot and its symmetric image with the minimum views move towards $m$. The equivalent robots can either move concurrently, or there may be a possible pending move. The configuration might become asymmetric due to the presence of a possible pending move. As a result, a multiplicity is created at some time instant $t >0$. Note that, in this case, if there exists a unique central node, then that remains the target during the execution of the algorithm. If there are two central nodes, the target might change during the execution of the algorithm, but the algorithm ensures that the target remains among $m_1$ and $m_2$. When a robot reaches one of $m_1$ and $m_2$, the target remains invariant during the execution of the algorithm.

\paragraph{Case 2} There is exactly one robot inside the intersecting rectangle $\mathcal {IR} (\mathcal R)$. Note that the robot must lie on the line of symmetry. The robot moves towards the target. While it moves towards the target, the configuration admits a unique target min-max node. Next, the robot and its symmetric image on the boundary, which is movable and has the minimum configuration view among all such robots, move towards $m$. Note that such robots always exist as it has been assumed that in the initial configuration $n \geq 9$. While they move towards an adjacent node, the robots become interior and the case proceeds similarly to that in Case 1.

\paragraph{Case 3} There are no robots inside the intersecting rectangle $\mathcal {IR} (\mathcal R)$. As it has been assumed that the initial configuration contains $n \geq 9$, then there must exist at least two robots on the boundary that are movable. Consider such a movable robot and its symmetric image with minimum views. Such robots move towards the target min-max node. While they move towards an adjacent node, the robots become interior and the procedure follows similar to as in Case 1. Note that while the robots move, the configuration may become asymmetric due to a possible pending move because of the asynchronous behaviour of the scheduler.

Next, assume the case when $G_M(\mathcal R)$ is a four-cycle graph. The following cases are to be considered.

\paragraph{Case 1} There are at least two robots inside the intersecting rectangle $\mathcal {IR}(\mathcal R)$. Consider the robots that are closest to the target and have the minimum views in case of ties. The robots move towards the target and a multiplicity is created at the target. Note the target might change when the robots move towards the target. However, the algorithm ensures that the target remains among $m_1$ or $m_2$. When the robots reach one of the chosen targets, the target remains invariant during the execution of the algorithm. Note that the configuration may become asymmetric due to a possible pending move of the robots.

\paragraph{Case 2} There is exactly one robot inside the intersecting rectangle $\mathcal {IR}(\mathcal R)$. The procedure follows similarly, as in the case when $G_M(\mathcal R)$ is a disconnected step-graph.

\paragraph{Case 3} There are no robots inside the intersecting rectangle $\mathcal {IR}(\mathcal R)$. The procedure follows similarly, as in the case when $G_M(\mathcal R)$ is a disconnected step-graph.

Next, assume the case when $G_M(\mathcal R)$ is a step-graph. The procedure follows similarly, as in the case when $G_M(\mathcal R)$ is a disconnected step-graph.

\subsubsection{Finalisation}
Suppose at time $t >0$, there exists a multiplicity at the target $m$. All the other robots that are not on $m$, move towards $m$ either synchronously or there might be a possible pending move due to the asynchronous behaviour of the scheduler. First, the interior robots that are not on $m$, move towards $m$. The robots that are closest to $m$ and minimum views in case of ties move towards $m$. As a result, at each instant of time, only two robots are allowed to move towards $m$. As the two robots lie on different half-planes, no other multiplicity would be created at any node other than $m$ while they move towards $m$. When all the interior robots reach $m$, the boundary robots that are closest to $m$ and minimum views in case of ties move towards $m$. While moving towards $m$, the equivalent robots first aligned along the same grid-line containing $m$ and then start moving towards $m$. As a result, the other robots can move towards $m$ without creating any other multiplicity other than $m$. Thus, the gathering is finalised at $m$.
\subsection{$\mathcal I_4$}
Consider the case when the initial configuration belongs to $\mathcal I_4$. This implies that $|V_M(\mathcal R)| \geq 2$ and the configuration is symmetric with respect to rotational symmetry. Assume that the center of rotational symmetry is $c$. Note that according to Lemma \ref{impossible3} and Lemma \ref{impossible4}, the configuration must be such that $V_m(\mathcal R)$ is a disconnected step-graph with an odd number of min-max nodes.
\subsubsection{Target Min-Max Selection Phase}
Consider the path induced by the set of min-max nodes. Since, the number of min-max nodes is odd, there always exists a unique central node (say $m$). The node $m$ is selected as the target min-max node. Note that $m$ is also the center of rotational symmetry $c$ of the configuration.
\subsubsection{Creating Multiplicity Phase}
Let $m$ be the target min-max node selected in the target min-max selection phase. Note that $m$ is also the unique Weber node in the configuration. As a Weber node always remains invariant during the movement of the robots towards itself, the target $m$ remains invariant during the movement of the robots. The four robots that are closest to $m$ and minimum views in case of ties, start moving towards $m$. As a result, a multiplicity is created at $m$ at some time instant $t>0$.
\subsubsection{Finalisation}
All the other robots that are not on $v$, start moving towards $m$ in the finalisation phase. Note that while the robots move towards $m$, $m$ remains the unique Weber node of the configuration. As a result, the target min-max node remains invariant. Thus, the gathering is finalised at $m$. Note that in the finalisation phase, a multiplicity may be created at the nodes other than $m$. However, since $m$ is the only Weber node of the configuration, the robots can identify $m$ as the unique target node, even when other multiplicities may exist in the configuration. 
\section{Correctness} \label{sec:correctness}
\begin{lemma} \label{lemma25}
    If $C(0) \in \mathcal I_1$, the target min-max node remains invariant in the Creating Multiplicity phase.
\end{lemma}
\begin{proof}
 If $C(0) \in \mathcal I_1$, there exists a unique min-max node $m$. We have to prove that $m$ remains the unique min-max node, while a multiplicity is created at $m$. That is, we have to prove that $m$ remains the min-max node and no more min-max nodes are created during the movement of the robots towards $m$. In order to prove this, we have to prove that during the movement of the robots towards $m$, the initial minimum enclosing diamond remains invariant. First, assume the case when the configuration is asymmetric. The following cases are to be considered.

 \paragraph{Case 1} Assume that the condition $(C1)$ from Lemma~\ref{lemma: condition for unique minmax point} holds. In the creating multiplicity phase, the algorithm ensures that there always exists at least one robot on every boundary of $\mathcal {MED (\mathcal R)}$. The $\mathcal {MED (\mathcal R)}$ remains invariant, while a multiplicity forms at node $m$ because, during the movement of the robots towards $m$, there is always at least one robot positioned on the boundary and only movable robots are allowed to move towards the target. As a result, no more min-max nodes are created during the movement and $m$ remains the unique min-max node. 

 \paragraph{Case 2} Assume that the condition $(C2)$ from Lemma~\ref{lemma: condition for unique minmax point} holds. In the creating multiplicity phase, the algorithm ensures that there always exists at least one robot on three boundaries of $\mathcal {MED (\mathcal R)}$, as only movable robots are allowed to move towards the targets. Furthermore, a robot on the side of the co-boundary corresponding to no boundary robot always exists during the formation of the multiplicity at $m$. Thus, $\mathcal {MED (\mathcal R)}$ remains invariant and hence $m$ remains the unique min-max node after the movement of the robots towards $m$.

 \paragraph{Case 3} Assume that the condition $(C3)$ from Lemma~\ref{lemma: condition for unique minmax point} holds. In the creating multiplicity phase, the algorithm ensures that there always exists at least one robot on two boundaries of $\mathcal {MED (\mathcal R)}$, as only movable robots are allowed to move towards the targets. Furthermore, robots on the side of the co-boundaries corresponding to no boundary robot always exist during the formation of the multiplicity at $m$. Thus, the proof follows.
 
 Next, assume the case when the configuration admits a single line of symmetry $l$. The target $m$ is selected as the unique min-max node on $l$. First, assume the case when every boundary of $\mathcal {MED (\mathcal R)}$ contain robots. The algorithm ensures that $\mathcal {MED (\mathcal R)}$ remains invariant as it only allows movable robots to move towards $m$. The case is similar when exactly two boundaries of $\mathcal {MED (\mathcal R)}$ contain robot positions. Thus, $\mathcal {MED (\mathcal R)}$ and $m$ remain invariant while no more min-max nodes are created. 

 Finally, consider the case when the configuration admits rotational symmetry. The target is selected as the node $c$, which is the center of rotational symmetry. Note that $c$ is also the unique Weber node of the configuration. While the robots move towards $c$, $c$ remains the unique Weber node of the configuration, and thus, the target remains invariant during the creating multiplicity phase. Thus, in the creating multiplicity phase, the target $m$ remains invariant. 
\end{proof}
\begin{lemma} \label{Lemma26}
    If $C(0) \in \mathcal I_1$, then in the creating multiplicity phase, a multiplicity is created at the target min-max node. 
\end{lemma}
\begin{proof}
    According to Lemma \ref{lemma25}, the target min-max node remains invariant in the creating multiplicity phase. We have to prove that a multiplicity is eventually created at some time $t>0$ in the creating multiplicity phase at $m$. First, assume that the configuration is asymmetric. The following cases are to be considered.

    \paragraph{Case 1} Assume that the condition $(C1)$ from Lemma~\ref{lemma: condition for unique minmax point} holds. If there are at least two interior robots, then robots can move towards the target in a sequential manner and a multiplicity is created at the target. If there exists at most one interior robot, then since the number of robots in the initial configuration is assumed to be at least nine, there exist movable robots on the boundaries. Thus, at least two movable robots are present in the configuration. The movable robots in the boundary move towards the interior and a multiplicity is created at the target.

    \paragraph{Case 2} Assume that the condition $(C2)$ from Lemma~\ref{lemma: condition for unique minmax point} holds. If the number of interior robots is at least two, then the procedure proceeds similarly as before. Otherwise, if the number of interior robots is exactly one, then it cannot move towards the target. If it does, then the target may change. However, since $n \geq 9$, there exist robots on the boundaries that are movable. Such movable robots move towards the interior and create a multiplicity at the target. 

    \paragraph{Case 3} Assume that the condition $(C3)$ from Lemma~\ref{lemma: condition for unique minmax point} holds. Since $n \geq 9$, there exists at least one robot on the boundaries that are movable. Consider such a movable robot that has the minimum view. The robot starts moving towards the target and the case proceeds similarly to that in Case 2.

    Next, consider the case when the configuration is symmetric with respect to a single line of symmetry $l$. As before, it should be noted that if $l$ is a horizontal or vertical line of symmetry, then there must exist robots at each boundary of $\mathcal 
    {MED(\mathcal R)}$. If there exist at least two interior robots, then the robots that are closest to $l$ and having minimum views in case of ties are allowed to move towards $m$. As a result, a multiplicity is created at $m$. Otherwise, if there exists only one interior robot, then that robot should lie on $l$. The robot can move towards $m$ along the line $l$. Next, the procedure follows similarly to the case when at least two interior robots exist. Thus, a multiplicity is created at $m$. Otherwise, if there exist no interior robots and since $n \geq 9$, there must exist boundary robots that are movable. Such movable robots that are closest to $m$ and having minimum views move towards $m$ and thus a multiplicity is created at $m$. In case $l$ is a diagonal line of symmetry, robots may exist on the co-boundaries. Consider the case when there exist at least two interior robots that are not lying on the co-boundaries. The procedure follows similarly to the case when the line of symmetry is either a horizontal or vertical line of symmetry. If a robot on a co-boundary is movable, then it moves towards $m$. Thus, the robot and its symmetric image create a multiplicity at $m$. Otherwise, the movable robots on the boundaries move towards $m$, and as a result, a multiplicity is created at the target in the creating multiplicity phase.

    Finally, assume the case when the configuration admits rotational symmetry. All the robots that are movable, move towards the center, and thus, a multiplicity is created at $m$. Hence, the proof follows.
\end{proof}
\begin{lemma} \label{lemma27}
    If $C(0) \in \mathcal I_1$, the gathering is finalised at the target min-max node.
\end{lemma}
\begin{proof}
    According to Lemma \ref{Lemma26}, a multiplicity is created at the target min-max node $m$. We have to first prove that in the finalisation phase, no other multiplicities are created in the finalisation phase. First, assume the case when the configuration is asymmetric. The robots are orderable. As a result, the robots that are not on $m$ can move towards $m$ in a sequential manner. Thus, while the robots move, no other multiplicities are created at $m$. Next, consider the case when the configuration is symmetric with respect to a single line of symmetry $l$. The robots that are closest to $m$ and having the minimum view move towards $m$. Note that the equivalent robots that first move towards $m$ lie on different half-planes. The robots first aligned themselves along the grid-line containing $m$ and then move towards $m$. As a result, no other multiplicities are created during the finalisation phase. In case the configuration admits rotational symmetry, the robots can create multiplicities at other nodes. However, since the target is the unique Weber node, the robots can finalise the gathering at $m$. As no other multiplicities are created during the finalisation phase, the robots can identify $m$ as the unique multiplicity node (robots have global weak-multiplicity detection capability). Thus, the robots that are not on $m$, move towards $m$ without creating any other multiplicity node, and thus the gathering is finalised at $m$. 
\end{proof}

 \begin{lemma} \label{disconnected1}
     Suppose $C(0) \in \mathcal I_2$ and the $G_M(\mathcal R)$ is a disconnected step-graph. Further, assume that $|V_M(\mathcal R)|$ is odd. Then, the target min-max node remains invariant in the creating multiplicity phase.
 \end{lemma}
\begin{proof}
    Since $G_M(\mathcal R)$ is a disconnected step-graph, the set of min-max nodes is a line graph, with the nodes lying on the diagonals of the minimum enclosing rectangle $\mathcal {MER}$. As $|V_M(\mathcal R)|$ is odd, the line graph possesses a unique central node (say $m$). The central node $m$ is selected as the target min-max node. Note that as a movable robot moves towards $m$, the path delimited by the set of min-max nodes becomes shorter in length accordingly. Since the path is subsequently becoming shorter in length, the size of the interior rectangle gradually diminishes in size and the central node remains the unique target min-max node of the configuration. As a result, at some instant of time $t>0$, a multiplicity is created at $m$, and the intersecting rectangle remains invariant. Thus, the target min-max node remains invariant in the creating multiplicity phase.
\end{proof}
\begin{lemma} \label{disconnected2}
    Suppose $C(0) \in \mathcal I_2$ and the $G_M(\mathcal R)$ is a disconnected step-graph. Further, assume that $|V_M(\mathcal R)|$ is even. Then, the algorithm ensures that during the movement of the robots towards the target in the creating multiplicity phase, no more target min-max nodes are created during the execution of the algorithm.
\end{lemma}
\begin{proof}
  If $|V_M(\mathcal R)|$ is even, then the target min-max node is selected among the two central nodes (say $m_1 $ and $m_2$). The following cases are to be considered.

  \paragraph{Case 1} There exists a unique Weber min-max node (say $m_1$). As a Weber min-max node remains invariant while a robot moves towards it and no more Weber min-max nodes are created during the movement of the robots towards it, $m_1$ remains the unique target min-max node. As a result, the target remains invariant while a multiplicity is created at it. Thus, no more min-max nodes are created during the execution of the algorithm.

  \paragraph{Case 2} The two min-max nodes $m_1$ and $m_2$ are both Weber min-max nodes. However, there exists a unique min-max node among $m_1$ and $m_2$, corresponding to which there exists a unique closest robot. Suppose $m_1$ is the unique Weber min-max node selected as a target and corresponds to which $r_1$ is the closest robot. While $r_1$ moves towards $m_1$, $r_1$ remains the closest robot to $m_1$, and the distance between $r_1$ and $m_1$ decreases gradually. It may be noted that during the movement of $r_1$ towards $m_1$, it may occur that the distance between $m_2$ and its closest robot also decreases. However, since $m_1$ is selected as the target initially, the distance between $m_1$ and $r_1$ is lesser than $m_2$ to its closest robot. As a result, $m_1$ remains the target min-max node during the movement of the closest robot to itself, and it remains the unique target while a multiplicity is created at it.

  \paragraph{Case 3:} The two min-max nodes $m_1$ and $m_2$ are both Weber min-max nodes. There does not exist a unique min-max node, corresponding to which there exists a unique closest robot. The min-max node $m_1$ is selected as the target, which has the minimum view among the possible targets. Let $r_1$ (resp. $r_2$) be the closest robot to $m_1$ (resp. $m_2$) and the corresponding distance between them is $d_1$ (resp. $d_2$). Now, $d_1=d_2$ in the initial configuration. Let at time $t>0$, $r_1$ starts moving towards the target $m_1$. Now, two cases may arise at $t>0$. Either $d_1 \neq d_2$ or $d_1=d_2$. If $d_1 \neq d_2$, the rest of the proof proceeds similarly as in Case 2. If $d_1=d_2$, then it may happen that $m_2$ becomes the min-max node with the minimum view between $m_1$ and $m_2$. It is, however, guaranteed that the target min-max node will be between $m_1$ and $m_2$ (since the target is always selected as the Weber min-max node, and the Weber min-max node remains invariant upon movement). Hence, the proof follows.
\end{proof}
\begin{lemma} \label{fourcycle}
    Suppose $C(0) \in \mathcal I_2$ and the $G_M(\mathcal R)$ is a four-cycle. Then, the algorithm ensures that during the movement of the robots towards the target, no more target min-max nodes are created in the creating multiplicity phase.
\end{lemma}
\begin{proof}
    In the initial configuration, the target min-max node is selected among the four nodes forming the four-cycle. Let the min-max nodes be denoted as $m_1$, $m_2$, $m_3$ and $m_4$. The following cases are to be considered.

    \paragraph{Case 1} There exists a unique Weber min-max node between the $m_i's$. The target is selected as the unique Weber min-max node, and the proof follows similarly as in Case 1 of Lemma \ref{disconnected2}.

    \paragraph{Case 2} There exists more than Weber min-max node between the $m_i's$. However, there exists a unique Weber min-max node (say $m_1$) corresponding to which the distance to the closest robot is minimum. The node $m_1$ is selected as the target min-max node, and the proof follows similarly as in Case 2 of Lemma \ref{disconnected2}.

    \paragraph{Case 3} The two min-max nodes $m_1$ and $m_2$ are both Weber min-max nodes. There does not exist a unique min-max node, corresponding to which there exists a unique closest robot. The min-max node $m_1$ is selected as the target, which has the minimum view between the possible targets. The rest of the proof follows similarly as in Case 3 of Lemma \ref{disconnected2}.
\end{proof}
\begin{lemma} \label{step1}
    Suppose $C(0) \in \mathcal I_2$ and the $G_M(\mathcal R)$ is a step-graph. Further, assume that $|V_M(\mathcal R)|$ is odd. Then, the target min-max node remains invariant during the execution of the algorithm.
\end{lemma}
\begin{proof}
     Since $G_M(\mathcal R)$ is a step-graph, the set of min-max nodes is a line graph, with the nodes lying on the three nodes of a unit square. As $|V_M(\mathcal R)|$ is odd, the line graph has a unique central node (say $m$). The central node $m$ is selected as the target min-max node. The rest of the proof follows similarly as in Lemma \ref{disconnected1}. 
\end{proof}
\begin{lemma} \label{step2}
    Suppose $C(0) \in \mathcal I_2$ and the $G_M(\mathcal R)$ is a step-graph. Further, assume that $|V_M(\mathcal R)|$ is even. Then, the algorithm ensures that during the movement of the robots towards the target, no more target min-max nodes are created during the execution of the algorithm.
\end{lemma}
\begin{proof}

If $|V_M(\mathcal R)|$ is even, then the target min-max node is selected among the two central nodes (say $m_1 $ and $m_2$). The rest of the proof follows similarly as in Lemma \ref{disconnected2}.
\end{proof}
\begin{lemma} \label{correctnessI2}
   If $C(0) \in \mathcal I_2$, then in the creating multiplicity phase, a multiplicity is created at the target min-max node. 
\end{lemma}
\begin{proof}
    Depending on whether $G_M(\mathcal R)$ is a disconnected step-graph, step-graph or a four-cycle, the following cases are to be considered.

    \paragraph{Case 1} $G_M(\mathcal R)$ is a disconnected step-graph. According to Lemmas \ref{disconnected1} and \ref{disconnected2}, the algorithm ensures that no more target min-max nodes are created during the execution of the algorithm. As a result, the movable robots either on the interior or on the boundary of $\mathcal{IR}(\mathcal R)$ move towards the target $m$, and thus a multiplicity can be created at $m$ at some time instant $t>0$.

    \paragraph{Case 2} $G_M(\mathcal R)$ is a four-cycle. According to Lemma \ref{fourcycle}, the algorithm ensures that no more target min-max nodes are created during the execution of the algorithm. The proof follows similarly from the previous case.

    \paragraph{Case 3} $G_M(\mathcal R)$ is a step-graph. According to Lemmas \ref{step1} and \ref{step2}, the algorithm ensures that no more target min-max nodes are created during the execution of the algorithm. The proof follows similarly from the previous case.
\end{proof}
\begin{lemma} \label{lemma33}
    If $C(0) \in \mathcal I_2$, then in the finalisation phase, the gathering is finalised at one of the target min-max nodes.
\end{lemma}
\begin{proof}
    According to Lemma \ref{correctnessI2}, a multiplicity is created at one of the target min-max nodes $m$. Since the configuration is asymmetric, the robots can be orderable. As a result, the robots that are not on $m$ can move towards $m$ without creating any other multiplicity. Since the robots have global weak multiplicity detection capability, all the robots can identify $m$ as the unique multiplicity node, and the gathering can be finalised at $m$.
\end{proof}
\begin{lemma} \label{lemma34}
     Suppose $C(0) \in \mathcal I_3$ and the $G_M(\mathcal R)$ is a disconnected step-graph. Then, the algorithm ensures that during the movement of the robots towards the target, no more target min-max nodes are created during the execution of the algorithm.
\end{lemma}
\begin{proof}
    Note that, in this case, the configuration can admit only a diagonal line of symmetry. Let $l$ be the line of symmetry. If there exists a unique min-max node on $l$, then note that the min-max node must be the central node of the path induced by the set of min-max nodes. The rest of the proof follows from Lemma \ref{disconnected1}. Otherwise, if there is more than one min-max node on $l$, then there must be either one or more than one central node of the subgraph induced by the set of the min-max nodes. The target is selected as the central node(s) of the subgraph induced by the set of min-max nodes. In case there is a unique central node, then the proof follows similarly to Lemma \ref{disconnected1}. Otherwise, the proof follows similarly to Lemma \ref{disconnected2}.
\end{proof}
\begin{lemma} \label{lemma35}
   Suppose $C(0) \in \mathcal I_3$ and the $G_M(\mathcal R)$ is a four-cycle. Then, the algorithm ensures that during the movement of the robots towards the target, no more target min-max nodes are created during the execution of the algorithm. 
\end{lemma}
\begin{proof}
    Note that in this case, there exist exactly two min-max nodes that lie on $l$. The target is selected as one of the min-max nodes on $l$. The rest of the proof follows similarly from Lemma \ref{fourcycle}.
\end{proof}

\begin{lemma} \label{lemma37}
    If $C(0) \in \mathcal I_3$, then in the creating multiplicity phase, a multiplicity is created at the target min-max node. 
\end{lemma}
    \begin{proof}
     Depending on whether $G_M(\mathcal R)$ is a disconnected step-graph, step-graph or a four-cycle, the following cases are to be considered.

    \paragraph{Case 1} $G_M(\mathcal R)$ is a disconnected step-graph. According to Lemma \ref{lemma33}, the algorithm ensures that no more target min-max nodes are created during the execution of the algorithm. As a result, the movable robots either on the interior or on the boundary of $\mathcal{IR}(\mathcal R)$ move towards the target $m$ and thus a multiplicity can be created at $m$ at some time instant $t>0$. The movable robots move towards $m$ only when it finds that it is the closest robot to $m$ and having the minimum views in case of ties.

    \paragraph{Case 2} $G_M(\mathcal R)$ is a step-graph. According to Lemma \ref{lemma34}, the algorithm ensures that no more target min-max nodes are created during the execution of the algorithm. The proof follows similarly from the previous case.

    \paragraph{Case 3} $G_M(\mathcal R)$ is a four-cycle. According to Lemma \ref{lemma35}, the algorithm ensures that no more target min-max nodes are created during the execution of the algorithm. The proof follows similarly from the previous case.
\end{proof}
\begin{lemma} \label{lemma38}
    If $C(0) \in \mathcal I_3$, then in the finalisation phase, the gathering is finalised at one of the target min-max nodes.
\end{lemma}
\begin{proof}
  According to Lemma \ref{lemma37}, a multiplicity is created at one of the target min-max nodes. We have to prove that no more multiplicities are created in the finalisation phase. All the other robots that are not on the target $m$, move towards $m$. First, the interior robots that are not on $m$ move towards $m$. The robots that are closest to $m$ and minimum views in case of ties move towards $m$. As a result, at each instant of time, only two robots are allowed to move towards $m$. As the two robots lie on different half-planes, no other multiplicity would be created at any node other than $m$ while they move towards $m$. While moving towards $m$, the equivalent robots first aligned along the same grid-line containing $m$ and then start moving towards $m$. As a result, the other robots can move towards $m$ without creating any other multiplicity other than $m$. Since the robots have global-weak multiplicity detection capability, all the robots can identify $m$ as the unique multiplicity node, and the gathering can be finalised at $m$. 
\end{proof}
\begin{lemma} \label{rotational}
    If $C(0) \in \mathcal I_4$, the target min-max node remains invariant during the execution of the algorithm. 
\end{lemma}
\begin{proof}
If $C(0) \in \mathcal I_4$, the target min-max node is selected as the min-max node $m$ on the center of rotational symmetry $c$. Due to the fact that $c$ is also the unique Weber node of the configuration, and a Weber node remains invariant during the robots' movement towards it, the node $c$ remains the target min-max node of the configuration. Thus, the proof follows.
\end{proof}
\begin{lemma} \label{rotational1}
  If $C(0) \in \mathcal I_4$, then in the finalisation phase, the robots can gather at the target min-max node.   
\end{lemma}
\begin{proof}
    According to Lemma \ref{rotational}, the target min-max node $m$ remains invariant during the execution of the algorithm. As a result, all the movable robots can move either synchronously or there may be a possible pending move due to the asynchronous behaviour of the scheduler towards $m$. Thus, a multiplicity is created at $m$. All the other robots that are not on $m$, can move towards $m$ and thus the gathering is finalised at $m$. Hence, the proof follows.
\end{proof}
\begin{lemma} \label{lemma41}
    Given $C(0) \in \mathcal {I \setminus U}$, where  $\mathcal I$ denote the set of all initial configurations and $\mathcal U$ denote the set of all ungatherable configurations. Let $t>0$ be an arbitrary instant of time at which at least one robot has completed its LCM cycle. During the execution of the algorithm Gathering (), $C(t) \notin \mathcal U$.
\end{lemma}
\begin{proof}
    The ungatherable configurations $\mathcal U$ are listed in the Corollaries \ref{corollary2}, \ref{corollary3} and \ref{corollary4}. The following cases are to be considered.
    
    \paragraph{Case 1} Assume the case when $C(0) \in \mathcal I_1$. The algorithm ensures that there exists a unique min-max node at each time $t>0$. Thus, $C(t) \in \mathcal I_1$ and hence, $C(t) \notin \mathcal U$. 
    
    \paragraph{Case 2} Assume the case, when $C(0) \in \mathcal I_2$. The algorithm ensures that the configuration remains asymmetric at any time $t>0$. Thus, $C(t) \in \mathcal I_2$ and hence $C(t) \notin \mathcal U $.

    \paragraph{Case 3} Assume the case, when $C(0) \in \mathcal I_3$. This implies that the configuration can admit only a diagonal line of symmetry. The algorithm ensures that while the robots move towards the target, the configuration either transforms into an asymmetric configuration or it may remain symmetric with respect to the same line of symmetry that was present in the initial configuration. In other words, the line of symmetry remains invariant during the execution of the algorithm. Thus, the configuration $C(t)$ at any time $t>0$, can not admit vertical or horizontal lines of symmetry. In fact, $C(t)$ can not admit a diagonal line of symmetry passing through the fixed lines of the intersecting rectangle if $C(0)$ does not admit such a kind of symmetry. Thus, $C(t) \notin \mathcal U$.

    \paragraph{Case 4} Assume the case, when $C(0) \in \mathcal I_4$. This implies that the configuration can be such that $V_M(\mathcal R)$ is a disconnected step-graph with an odd number of min-max nodes. The algorithm ensures that while the robots move towards the target, the configuration either transforms into an asymmetric configuration or it may remain symmetric with respect to rotational symmetry. If it transforms into an asymmetric configuration, then $C(t) \in \mathcal U$, for any $t >0$. Otherwise, while the interior robots move towards the target, $V_M(\mathcal R)$ remains a disconnected step-graph with an odd number of min-max nodes. Thus, while the robots move towards the target, the path of min-max nodes gradually decreases in length, and the number of min-max nodes becomes one. Hence, $C(t) \in \mathcal U$, for any $t >0$.
\end{proof}

\begin{theorem}
    If the initial configuration belongs to the set $\mathcal{I} \setminus \mathcal{U}$, then algorithm Gathering() ensures gathering over min-max nodes.
\end{theorem}
\begin{proof}
    Assume that $C(0) \in \mathcal{I} \setminus \mathcal{U}$. If $C(t)$ is not a final configuration for some $t \geq 0$, each active robot executes the algorithm Gathering(). According to the Lemmas \ref{lemma41}, any initial configuration $C(0) \in \mathcal{I} \setminus \mathcal{U}$, would never reach
a configuration $C(t) \in \mathcal U$, at any point of time $t > 0$ during the execution of the algorithm Gathering(). The following cases are to be considered.

\paragraph{Case 1 } $C(0) \in \mathcal I_1$. According to Lemma \ref{lemma27}, the gathering is finalised at the target min-max node.

\paragraph{Case 2} $C(0) \in \mathcal I_2$. According to Lemma \ref{lemma33}, the gathering is finalised at the target min-max node.

\paragraph{Case 3} $C(0) \in \mathcal I_3$. According to Lemma \ref{lemma38}, the gathering is finalised at the target min-max node.

\paragraph{Case 4} $C(0) \in\mathcal I_4$. According to Lemma \ref{rotational}, the gathering is finalised at the target min-max node.

\noindent Thus, for any initial configuration, $C(0) \in \mathcal{I} \setminus \mathcal{U}$, gathering is finalised at one of the min-max nodes and within a finite time by the execution of the algorithm Gathering().
\end{proof}
\section{Conclusion and Future Works} \label{sec:conclusion}
In this paper, we address the min-max gathering problem for oblivious robots on an infinite grid. The aim of the problem is to minimize the maximum distance that any robot must travel during the gathering process. In the continuous domain, it can be proved that the min-max point is unique. However, we have seen examples where there can be more than one min-max node if the robots are deployed at the nodes of an infinite grid. We have characterized the subgraph induced by the set of min-max nodes and observed that it may be either a step-graph, a disconnected step-graph or a four-cycle. All those initial configurations for which gathering under objective constraint cannot be ensured have also been characterized. For the remaining configurations, a deterministic distributed algorithm has been proposed to ensure that the robots are gathered under the objective constraint.

One of the future directions to work with can be to consider the problem in the presence of meeting nodes. In \cite{DBLP:journals/ijfcs/BhagatCDM23}, the min-sum gathering problem with meeting nodes was studied, where the goal was to gather the robots at one of several predefined fixed nodes on the grid, minimizing the total distance traveled by all robots. An extension of this could focus on the min-max gathering problem with meeting nodes, where the objective is to gather the robots at one of the meeting nodes while minimizing the maximum distance traveled by any individual robot. Another direction for future research would be to consider the problem without any multiplicity detection capability. As certain initial symmetric configurations prevent gathering at min-max nodes, exploring randomized algorithms for these configurations could be an interesting direction for future research.
\bibliographystyle{elsarticle-num}
\bibliography{main}


\end{document}